\documentclass[a4paper,11pt]{llncs}
\usepackage[utf8]{inputenc}
\usepackage[english]{babel}
\usepackage{amssymb, amsmath, amsfonts}

\usepackage[font={small,it}]{caption}

\usepackage{microtype}

\usepackage[utf8]{inputenc}
\usepackage{hyperref}

\usepackage{tikz}

\usetikzlibrary{positioning} 
\usetikzlibrary{shapes} 
\usetikzlibrary{arrows,arrows.meta,decorations.markings} 
\tikzstyle directed=[postaction={decorate,decoration={markings,mark=at position .65 with {\arrow[scale=1.2]{>}}}}]

\usepackage{color} 
\usepackage{ulem} 
\newcommand{\flo}[1]{\textcolor{black}{#1}}

\newcommand{\xor}{\oplus}

\newcommand{\erase}[1]{}

\renewcommand{\bar}{\overline}

\newcommand{\letterA}{0}
\newcommand{\letterB}{1}
\newcommand{\letterC}{2}
\newcommand{\letterD}{3}

\DeclareMathOperator{\GNECC}{\mathsf{NECC}}

\DeclareMathOperator{\ANs}{ANs}
\DeclareMathOperator{\AN}{AN}
\DeclareMathOperator{\IG}{IG}
\DeclareMathOperator{\Pw}{Pw}
\DeclareMathOperator{\pr}{pr}

\begin{document}

\title{
	Sequentialization and Procedural Complexity in Automata Networks
	}

\author{
	Florian Bridoux\inst{1}\thanks{
	 \url{florian.bridoux@lis-lab.fr}.}
}

\date{\today}

\institute{
	Aix-Marseille Univ., Toulon Univ., CNRS, LIS, Marseille, France
}  

\maketitle

\begin{abstract}
	
	In this article we consider finite automata networks ($\ANs$) with two kinds of update 
	schedules:
	the parallel one (all automata are updated all together) and the sequential ones 
	(the automata are updated periodically one at a time according to a total 
	order $w$). 
	\flo{
	The cost of sequentialization of a given $\AN$ $h$ is the number of additional automata required to simulate $h$ by a sequential $\AN$ with the same alphabet.
	We construct, for any n and q, an $\AN$ $h$ of size $n$ and alphabet size $q$ whose cost of sequentialization is at least $n/3$.
	We also show that, if $q \geq 4$, we can find one whose cost is at least $n/2-\log_q(n)$.
	We prove that $n/2 + 
	\log_q(n/2+1)$ is an upper bound for the cost of sequentialization of any $\AN$ $h$ of size $n$ and alphabet size $q$.
	Finally, we exhibit the exact relation between the cost of sequentialization of $h$ and its procedural complexity with unlimited memory and prove that its cost of sequentialization is less than or equal to the pathwidth of its interaction graph.
	}
\end{abstract}\setcounter{lemma}{0}


\section{Introduction}\label{section_intro}
\setcounter{lemma}{2}
In this article, we study finite automata networks ($\ANs$).
They are models classically used for representing and analyzing natural dynamical systems like genetic or neural networks~\cite{Thomas1973,Goles1990}.
Moreover, they are also computational models on which we study computability and complexity properties which is the purpose of this paper.
\flo {An $\AN$ $h$ can be seen as a transformation of $A^n$ with $A$ a finite alphabet. Here, $n$ is the number of automata, and the $i$-th component of $h$ is the update function of the 
$i$-th automaton.}
We consider them with two types of update schedules.
With the parallel one, automata are updated all together, at each time step. In other words, we just apply $h$. 
\flo{ With the sequential ones, automata are updated sequentialy, according to a total order $w$}. 
\flo{They have been several works on the influence of the update schedules on the function computed by an $\AN$~\cite{Goles2008,Aracena2009}. Here, like in~\cite{Tchuente1986} we take the opposite approach. We have an $\AN$ $h$ with a parallel update schedule and try to find an $\AN$ $f$ with a sequential update schedule $w$ which computes the same function.}
\flo{However, sometime it is impossible.} For instance, the transformation of $\{0,1\}^2$ which exchanges the two values $h:(x_1,x_2) 
\mapsto (x_2,x_1)$ cannot be sequentialized.
The famous XOR swap algorithm, 
$x_1 \leftarrow x_1 \xor x_2$,
$x_2 \leftarrow x_1 \xor x_2$,
$x_1 \leftarrow x_1 \xor x_2$
does not apply here because we can only update \flo{ one time each automaton beetween two time steps}.
However, what we can do is to consider the $\AN$ $f$ with one additional automaton and the sequential update schedule $w := (3,2,1)$ which executes the three instructions 
$x_3 \leftarrow x_1$, 
$x_1 \leftarrow x_2$, 
$x_2 \leftarrow x_3$.
We see that $f$ with the update schedule $w$ computes the transformation $h$ if we only consider the $2$ first automata.
The goal of this paper is to determine the cost of sequentialization of an $\AN$ $h$, namely, the minimum number of additional automata that an $\AN$ $f$ which sequentializes $h$ will have.
This paper is the direct sequel of~\cite{Bridoux0PST17} in which the same problem was studied for an alphabet of size $2$ and with an imposed order of sequentialization.
Definition~\ref{def_confusion_graph}, Theorem~\ref{th1} and Lemma~\ref{lem_kp_lower_bound} are straightforward generalization of results published in~\cite{Bridoux0PST17}.
All other results are new.

\flo{In Section~\ref{section_def}, we define $\ANs$, interaction graphs, the notion of a sequentialization and we present most of the notations that we use. 
In Section~\ref{section_added_automata}, we define the cost of sequentialization $\kappa(h,u)$ of an $\AN$ $h$ respecting an order $u$. It is the minimum number of additional automata required for any $\AN$ $f$ with a sequential update schedule $w$ respecting the order $u$ to compute $h$. We also define $\kappa^{\min}(h)$ which is like $\kappa(h,u)$ except that the sequential update schedules we consider are not constraint anymore.
In Section~\ref{section_confusion_graph}, 
we give an upper and lower bounds for $\kappa(h,u)$ for the couple $(h,u)$ which maximizes it. 
In Section~\ref{section_bounds_kmin}, we prove different lower bounds depending on the alphabet size for $\kappa^{\min}(h)$ when $h$ maximizes $\kappa^{\min}(h)$.
In Section~\ref{section_p_c} we give the relation between $\kappa^{\min}$ and the procedural complexity as defined in \cite{Gadouleau2011}.
Finally, In Section~\ref{section_bound_path}, we prove an upper bound for $\kappa^{\min}(h)$ depending on the pathwidth of the interaction graph of $h$.}

\section{Definitions and notations} \label{section_def}


For all $i \in \mathbb{N},$ the interval between $1$ and $i$ is denoted by $[i] := \{1,2,\dots, i \}$.
For all $i,j \in \mathbb{N}$, with $i \leq j$, the closed interval between $i$ and $j$ is denoted by $[i,j] := \{i,i+1,\dots,j\}$ and the open one by $]i,j[ := [i,j] \setminus \{i,j\}$.
For any $q \geq 2$ and $n \in \mathbb{N}$, let $F(n,q)$ be the set of functions from $[0,q[^n$ to $[0,q[^n$ (also called transformations of $[0,q[^n$).
For all $I = \{i_1,i_2,\dots, i_p\} \subseteq [n]$ with $i_1 < i_2 < \dots < i_p$, the projection of $x$ on $I$ is denoted either by $ \pr_I(x)$ or by $x_I$. In other words, $\pr_I(x) = x_I = (x_{i_1},x_{i_2},\dots,x_{i_p})$. 
For all vectors $x := (x_1,\dots,x_p)$ and $y := (y_1,\dots,y_t)$, their concatenation is denoted by $xy := (x_1, \dots, x_p, y_1, \dots y_t)$.
\begin{definition}[Coordinate functions]
	Let $f \in F(n,q)$.
	\flo{For every $i \in [n]$, the $i$-th coordinate functions of $f$ is the function $f_i := \pr_{i} \circ f$}.
\end{definition}
This means that we have $f(x) = (f_1(x),f_2(x),\dots,f_n(x))$.
In this paper, we make particular use of the superscript of a function $f$.
\begin{definition}[Updates of a transformation]
		For all $i \in [n],$ $f^i \in F(n,q)$ is the function which updates the $i$-th coordinate (\em i.e. \em executes $f_i$).
		For all $I \subseteq [n]$, $f^I$ is the function which updates the coordinates of all elements of $I$ synchronously.
		For any word $w:= (w_1,w_2, \dots,w_t)$ on the alphabet $[n]$, $f^w$ is the function which updates sequentially the coordinates $w_1, \dots, w_t$ in the order given by $w$.
\end{definition}
Formally, we have
	\begin{equation*}
	\begin{array}{ll}
\forall x \in A^n,& f^i(x) := (x_1,\dots,x_{i-1},f_i(x),x_{i+1},\dots,x_n).\\
\forall x \in A^n,j \in [n],\ & f^I(x)_j := \begin{cases}
 f_j(x) &\text{if } j \in I \\
 x_j &\text{otherwise}.
\end{cases}\\
\forall w = (w_1,w_2, \dots,w_t) \in {[n]}^t,& f^w := f^{w_t}\circ \dots \circ f^{w_2} \circ f^{w_1}.\\
	\end{array}
	\end{equation*}
We say that $f_i$ is a trivial coordinate function if for all $x \in A^n, f_i(x) = x_i$.
The relation $y = f^i(x)$ can be expressed by $x \xrightarrow{f^{i}} y$.
The set of permutations of $[n]$ 
is denoted by $\Pi([n])$.
Let $w:= (w_1,w_2, \dots,w_t) \in \Pi([n])$.
If $w_j = i$ then we say that $i$ is updated at step $w(i) := j$.
\begin{definition}[Sequentialization]
	An $\AN$ $f \in F(m,q)$, with the sequential update schedule $w \in \Pi([m])$ sequentializes an $\AN$ $h \in F(n,q)$ with $m \geq n$ if $\pr_{[n]} \circ f^w = h \circ pr_{[n]} $.
\end{definition}

\begin{remark}
	All the results of this paper remain true if we use the more general definition: 
	$ \exists I \subseteq [m],$ with $|I| = n$ such that $\pr_{I} \circ f^w = h \circ pr_{I}$.
\end{remark}

\begin{definition}[Interaction graph]
	The interaction graph $\IG(h)$ of an $\AN$ $h \in F(n,q)$ is the directed graph $([n],E)$ with $(i,j) \in E$ if and only if $i$ has an influence on $j$.
	More formally, $\forall i,j \in [n], (i,j) \in E$ if and only if $\exists x,y \in A^n$ such that $ x_{[n] \setminus \{i\}} = y_{[n] \setminus \{i\}}$ and $ h_j(x) \neq h_j(y)$. 
	
\end{definition}

\flo{We denote by $\IG^*(h)$ be the undirected version of $\IG(h)$.}

\section{Cost of sequentialization } \label{section_added_automata} 

\flo{In this section, we define the main question tackled in this paper.}
For all $u \in \Pi([n])$, we say that $w \in \Pi([m])$ respects $u$, if all the coordinates of $[n]$ are updated in the same order in $u$ and in $w$.
In other words, $\forall i, j \in [n],$ if $u(i) < u(j)$ then $w(i) < w(j)$.

\begin{definition}[$\kappa(h,u)$]
	Let $h \in F(n,q)$ and $u \in \Pi([n])$.
	The cost of sequentialization of $h$ respecting $u$, denoted by $\kappa(h,u)$, is the smallest $k$ such that there exists $f \in F(n+k,q)$ and $w \in \Pi([n+k])$, such that $(f,w)$ sequentializes $h$ and $w$ respects $u$.
\end{definition} 
\begin{definition}[$\kappa^{\min}(h)$] \label{def_k}
		Let $h \in F(n,q)$.
		The cost of sequentialization of $h$, denoted by $\kappa^{\min}(h)$, is the smallest $k$ such that there is a $f \in F(n+k,q)$ and a $w \in \Pi([n+k])$, such that $(f,w)$ sequentializes $h$.
\end{definition}

Clearly, $\kappa^{\min}(h) = \min(\{ \kappa(h,u)\ |\ u \in \Pi([n]) \})$.
Given $n$ and $q$, the maximal cost of sequentilization respectively with or without imposed order is denoted by $\kappa_{n,q} := max(\{ \kappa(h,u)\ |\ h \in F(n,q)$ and $u \in \Pi([n]) \})$ and $\kappa^{\min}_{n,q} := max(\{ \kappa^{\min}(h)\ |\ h \in F(n,q) \})$, respectively.
Example~\ref{k_vs_km} shows that, for some $(h,u)$, the difference between $\kappa^{min}(h)$ and $\kappa(h,u)$ is large.

\begin{example} \label{k_vs_km}

		\begin{figure}[!h]
			
\minipage[t]{0.32\textwidth}

			\centering
			\begin{tikzpicture}[scale=0.8]
			\tikzstyle{automate} = [draw, circle]
			\tikzstyle{grossefleche} = [-{>[length=1.5mm]}]
			\tikzstyle{ehcelfessorg} = [{<[length=1.5mm]}-]
			
			\node[automate](n0) at (0,3) {$1$};
			\node at (0,3.5) {$x_4$};
			\node[automate](n1) at (0,1.5) {$2$};
			\node at (0,2) {$x_5$};
			\node[automate](n3) at (0,0) {$3$};
			\node at (0,.5) {$x_6$};
			
			\node[automate](n4) at (3,3) {$4$};
			\node at (3,3.5) {$x_1$};
			\node[automate](n5) at (3,1.5) {$5$};
			\node at (3,2) {$x_2$};
			\node[automate](n6) at (3,0) {$6$};
			\node at (3,.5) {$x_3$};
			
			\draw[ehcelfessorg] (n0) edge[bend left] (n4);
			\draw[ehcelfessorg] (n1) edge[bend left] (n5);
			\draw[ehcelfessorg] (n3) edge[bend left] (n6);
			
			\draw[ehcelfessorg] (n4) edge[bend left] (n0);
			\draw[ehcelfessorg] (n5) edge[bend left] (n1);
			\draw[ehcelfessorg] (n6) edge[bend left] (n3);
			
			\end{tikzpicture}
			
			\caption{Interaction graph of the $\AN$ $h$ of Example~\ref{k_vs_km}.}\label{figure_example_1_h}

\endminipage\hfill
\minipage[t]{0.32\textwidth}
		\centering
		\begin{tikzpicture}[scale=0.8]
		\tikzstyle{automate} = [draw, circle]
		\tikzstyle{grossefleche} = [-{>[length=1.5mm]}]
		
		\node[automate](n0) at (0,3) {$1$};
		\node at (0,3.5) {$z_4$};
		\node[automate](n1) at (0,1.5) {$2$};
		\node at (0,2) {$z_5$};
		\node[automate](n3) at (0,0) {$3$};
		\node at (0,.5) {$z_6$};
		
		\node[automate](n4) at (2,3) {$4$};
		\node at (2,3.5) {$z_7$};
		\node[automate](n5) at (2,1.5) {$5$};
		\node at (2,2) {$z_8$};
		\node[automate](n6) at (2,0) {$6$};
		\node at (2,.5) {$z_9$};
		
		\node[automate](n7) at (4,3) {$7$};
		\node at (4,3.5) {$z_1$};
		\node[automate](n8) at (4,1.5) {$8$};
		\node at (4,2) {$z_2$};
		\node[automate](n9) at (4,0) {$9$};
		\node at (4,.5) {$z_3$};
		
		\draw[grossefleche] (n0) edge[bend left] (n7);
		\draw[grossefleche] (n1) edge[bend left] (n8);
		\draw[grossefleche] (n3) edge[bend left] (n9);
		\draw[grossefleche] (n7) -- (n4);
		\draw[grossefleche] (n8) -- (n5);
		\draw[grossefleche] (n9) -- (n6);
		\draw[grossefleche] (n4) -- (n0);
		\draw[grossefleche] (n5) -- (n1);
		\draw[grossefleche] (n6) -- (n3);
		\end{tikzpicture}
		
		\caption{Interaction graph of the $\AN$ $f$ of Example~\ref{k_vs_km}.}
	\label{figure_example_1_f}
		
	\endminipage\hfill
\minipage[t]{0.32\textwidth}
		
		\centering
		
		\begin{tikzpicture}[scale=0.8]
		\tikzstyle{automate} = [draw, circle]
		\tikzstyle{grossefleche} = [-{>[length=1.5mm]}]
		
		\node[automate](n1) at (2.5,6) {$1$};
		\node[automate](n2) at (2.5,4) {$2$};
		\node[automate](n3) at (2.5,2) {$3$};
		\node at (2.5,6.75) {$y_4$};
		\node at (2.5,4.75) {$y_5$};
		\node at (2.5,2.75) {$y_6$};
		
		\node[automate](n4) at (4,6) {$4$};
		\node[automate](n5) at (4,4) {$5$};
		\node[automate](n6) at (4,2) {$6$};
		
		\node at (5.25,6.75) {$y_7 - y_2 - y_3$};
		\node at (5.25,4.75) {$y_7 - y_4 - y_3$};
		\node at (5.25,2.75) {$y_7 - y_4 - y_5$};
		
		\node[automate](n) at (5.5,4) {$7$};
		\node at (5.5,3.5) {$y_1 + y_2 + y_3 $};
		
		\draw[grossefleche] (n) -- (n5);
		\draw[grossefleche] (n3) -- (n5);
		\draw[grossefleche] (n4) -- (n5);
		
		\end{tikzpicture}
		\caption{Interaction graph of the $\AN$ $g$ of Example~\ref{k_vs_km} with only inner edges of the automaton $5$ displayed.
		}\label{figure_example_1_g}

	\endminipage\hfill
\end{figure}
Let us consider the $\AN$ $h \in F(n,q)$ with $n=6$ which computes the swaps of the values of $3$ pairs of automata.
In other words, $$h: x \mapsto (x_4,x_5,x_6,x_1,x_2,x_3).$$
Figure~\ref{figure_example_1_h} displays the interaction graph of $h$. Now, we consider the canonical sequential update schedule $u = (1,2,\cdots,6)$
and we want to find an $\AN$ $f$ and a update schedule $w$ which sequentializes $h$ respecting $u$.
To do so, let us consider an $\AN$ $f \in F(9,2)$ and $w \in \Pi([9])$.
First, we define the order $w:=(7,8,9,1,2,3,4,5,6)$ which updates the $n/2$ additional automata of $f$ before it updates the $n$ first ones.
Then, we take $f$ which copies the values of the first set of automata in the third, the second in the first and the third in the second. Formally,
$f: z \mapsto z_{[4,6]} z_{[7,9]} z_{[3]}$.
Figure~\ref{figure_example_1_f} shows the interaction graph of $f$.
Now, a simple expansion of $f^w$ gives us
$$ z = z_{[3]} z_{[4,6]} z_{[7,9]} \xrightarrow{f^{7,8,9}} z_{[3]} z_{[4,6]} z_{[3]} \xrightarrow{f^{1,2,3}} z_{[4,6]} z_{[4,6]} z_{[3]} \xrightarrow{f^{4,5,6}} 
 h(z_{[6]}) z_{[3]}.$$
Thus, we have $\pr_{[n]} \circ f^w = h \circ \pr_{[n]}$.
As a result, $(f,w)$ sequentializes $h$ respecting $u$ and $\kappa(h,u) \leq 3$.
Moreover, Lemma~\ref{lem_kp_lower_bound} 
(Section \ref{section_confusion_graph}), 
shows that there are no smaller $(f,w)$ which would suit.
Thus, we have $\kappa_{6,q} \geq \kappa(h,u) = n/2 = 3$.
Next, we define $g \in F(7,2)$ and $v \in \Pi([7])$ such that $(g,v)$ (with only one more automaton than $h$) sequentializes $h$ (but without respecting $u$). 
First, we define the order $v := (7,1,4,2,5,3,6)$ which, instead of updating $[n]$ in the order $u$, updates the pairs of automata $(1,4)$, $(2,5)$ and $(3,6)$ one by one.
Then, we take $g$ such that for all $y \in \{0,1\}^7$, 
$$g: y \mapsto (y_4,y_5,y_6,y_7 - y_2 - y_3,y_7 - y_4 - y_3, y_7 - y_4 - y_5,y_1+y_2+y_3).$$
Figure~\ref{figure_example_1_g} depicts the interaction graph of $g$ with only the inner edges of the automaton $5$ displayed.
As above, a simple expansion of $g^v$ gives us $g^v:y \mapsto  h(y_{[6]})(y_1+y_2+y_3)$.
Thus, $\pr_{[n]} \circ g^v = h \circ \pr_{[n]}$ and $g$ has $1$ more automata than $h$.
As a result, $(f,w)$ sequentializes $h$ and $\kappa^{\min}(h) \leq 1$. A generalization of this example shows that for all even $n$ and $q \geq 2,$ $\exists h \in F(n,q), u \in \Pi([n])$ such that $\kappa(h,u)  \geq \kappa^{\min}(h) + n/2 -1$.
\end{example}

\section{Confusion graph and $\kappa_{n,q}$} \label{section_confusion_graph} 

 \flo{In \cite{Bridoux0PST17}, the $\GNECC$ graph was defined. 
 This graph is very useful to compute $\kappa(h,u)$.
 We rather call it the confusion graph in this paper. }

\begin{definition}[Confusion graph] \label{def_confusion_graph}
	\flo{Let us consider $h \in F(n,q)$ and the sequential update schedule $u \in \Pi([n])$. 
	We call \emph{confusion graph} $G_{h,u}$ the undirected graph whose vertices are all the configurations of $[0,q[^n$ and in which two configurations $x$ and $x'$ are neighbors if and only if $h(x) \neq h(x')$ and $\exists\ i \in [n],$ $h^{ \{u_1, \dots, u_i \}}(x) =  h^{\{u_1, \dots, u_i \}}(x')$.}
\end{definition}

In the sequel, we denote by $\chi(G)$ the \emph{chromatic number} of the graph $G$, namely the minimum number of colors of a proper coloring of its vertices.
In \cite{Bridoux0PST17}, the exact relation between the chromatic number of the confusion graph $G_{h,u}$ and $\kappa(h,u)$ was proven in the case where $q=2$.
We propose in Theorem~\ref{th1} a straightforward generalization for any alphabet size.

\begin{theorem} 
	\label{th1}
	Let us consider $h \in F(n,q)$ and the sequential update schedule $u \in \Pi([n])$. 
	Then we have $\kappa(h,u) = \lceil \log_q(\chi(G_{h,u})) \rceil$.
\end{theorem}

%
%
%


\flo{
In \cite{Bridoux0PST17},the authors proved that for all $n$ we can construct $h \in F(n,2)$ whose cost of sequentialization respecting the order $u \in \Pi([n])$ is $\lfloor n/2 \rfloor$. Lemma~\ref{lem_kp_lower_bound} bellow is a straightforward generalization for any alphabet size.}

\begin{lemma} 
	\label{lem_kp_lower_bound}
	For all $ n \in \mathbb{N}$ and $q \geq 2$ we have $\kappa_{n,q} \geq \lfloor n/2 \rfloor$.
\end{lemma}

Moreover, in \cite{Bridoux0PST17}, the authors showed that $\forall n \in \mathbb{N}, \kappa_{n,2} \leq 2n/3+2$.
Theorem~\ref{th2} below shows that we have in fact, $\kappa_{n,q} \leq \lceil n/2 + \log_q(n/2+1) \rceil$ for any $q$.
To prove it, we regroup all the configurations of the confusion graph $G_{h,u}$ which are equal in their second half $(x_{\{u_{n/2+1}, \dots, u_{n}\}} = x'_{\{u_{n/2+1}, \dots, u_{n}\}})$ and have the same image $(h(x) = h(x'))$.
We prove that a proper coloring of this graph is a proper coloring of the confusion graph. And then, we prove that the maximal degree of this factorized graph is at most $\lceil (n/2+1)q^{n/2} \rceil$. Since the chromatic number of a graph is at most its maximal degree (plus one), we deduce an upper bound for the chromatic number and then for $\kappa_{n,q}$. 
\begin{theorem}
	\label{th2}
	For all $ n \in \mathbb{N},\ q \geq 2$ we have $\kappa_{n,q} \leq \lceil n/2 + log_q(n/2+1) \rceil$.
\end{theorem}

\section{Lower bounds for $\kappa^{\min}_{n,q}$} \label{section_bounds_kmin}

The goal of this section is to construct an $\AN$ with the biggest cost of sequentialization possible and thus deduce a lower bound for $\kappa^{\min}_{n,q}$.
For any set $I$, the set of subsets of $I$ of size $k$ is denoted by $\binom{I}{k} := \{ J \subseteq I\ |\ |J|=k \}$.
For all $x \in A^n$ and $I \subseteq [n]$, let $x[I]:= \{ x' \in A^n\ |\ {x'}_{[n] \setminus I } = x_{[n] \setminus I } \}$ be the set of configurations of $A^n$ which only differ from $x$ in $I$.
In Lemma~\ref{lem_a_b}, we prove that if we can find an encoding $b:\binom{[2k]}{k} \to A^n$ such that the sets $b(E)[E]$ with $E \in \binom{[2k]}{k}$ are disjoint, then there exists $h \in F(n,q)$ such that $\kappa^{min}(h) \geq k$.
To do so, we define the function $h$ \flo{such that for all $x \in b(E)[E],\ h_{E}(x) = x_{[2k] \setminus E}$ and $h_{[2k]\setminus E}(x) = x_{E}$}.
For any $u \in \Pi([n])$, we can define $E$ as the $k$ first coordinates updated by $u$ in $[2k]$ and consider $x = b(E)$.
The set $x[E]$ is a clique in the confusion graph $G_{h,u}$. Indeed, any function which sequentializes $h$ respecting $u$ has, for any configuration in $x[E]$, to first erase the information in $E$ and then to restore it in $[2k]\ \setminus\ E$.
Since this clique is of size $q^k$, we have $\kappa_{h,u} \geq k$ for any $u$ and $\kappa^{\min}(h) \geq k$.

\begin{lemma} \label{lem_a_b}
	Let $n, k \in \mathbb{N} $ and $q \geq 2$.
	If there is a function $b: \binom{[2k]}{k} \to  [0,q[^n$ such that the sets $b(E)[E]$ with $E \in \binom{[2k]}{k}$ are disjoint then there exists a $h \in F(n,q)$ without trivial coordinate functions, with $\kappa^{min}(h) \geq k$.
\end{lemma}

Using Lemma~\ref{lem_a_b} we could easily show that for any $q \geq 2$ and $n \in \mathbb{N}$, we have $\kappa^{min}_{n,q} \geq \lfloor n/4 \rfloor $.
Indeed, if we have $n = 4k$, we can use the second half of the configuration to encode the set $E$. In Theorem~\ref{th_km_2} we prove that for any alphabet, we can in fact encode any $E \in \binom{[2k]}{k}$ in a configuration $x$ of size $3k$.
To do so, we use the following technique:
if $i \in \bar{E} := [2k] \setminus E$ then we have $x_i = \letterA$ if $i+1$ in $E$ and $\letterB$ otherwise. Moreover, in $[2k+1,3k]$, using the same technique, we indicate if each element of $E$ is followed by another element of $E$ or not.
From this encoding and Lemma~\ref{lem_a_b} we deduce a lower bound for $\kappa^{\min}$ for any alphabet.
\begin{theorem} \label{th_km_2}
	For all $n \in \mathbb{N}$ and $q \geq 4$, we have $\kappa^{\min}_{n,q} \geq \lfloor n/2 - \log_q(n)\rfloor$.
\end{theorem}

Theorem~\ref{th_km_s} below states that, if we have an alphabet of size at least $4$, we can encode any $E \in \binom{[2k]}{k}$ in a configuration of size $2k+log_q(2k)$.
To do so, we encode $E$ in $[2k] \setminus E$ using the fact that in an alphabet of size $4$ each coordinate can encode twice more information than with a bit.
Then, we indicate in $[2k,2k+log_q(2k)]$ where the reading for decoding starts. From this encoding and Lemma~\ref{lem_a_b} we deduce a lower bound for $\kappa^{\min}$.
\begin{theorem} \label{th_km_s}
	For all $q \geq 4 , n \in \mathbb{N}$, $\kappa^{\min}_{n,q} \geq \lfloor n/2 - \log_q(n)\rfloor$.
\end{theorem}

\section{Procedural complexity} \label{section_p_c}

Now, we study the relation between $\kappa^{\min}$ and the procedural complexity as defined in \cite{Gadouleau2011}.
The procedural complexity of $h$ is the minimum number $t$ of functions $g^{(1)}, \dots, g^{(t)}$ (each of which update at most one coordinate) that are required for $g^{(t)} \circ \dots \circ g^{(1)}$ to compute $h$.
For all $q \geq 2$ and $n \geq 2$, let us denote by $F^*(n,q) \subseteq F(n,q)$ the set of functions which do not update more than one coordinate.
In~\cite{Gadouleau2011}, the authors first studied the memoryless procedural complexity $\mathcal{L}(h)$.
It is the necessary number of step to compute $h$ with $g^{(1)}, \dots, g^{(t)}$ of same size than $h$.
Then, they studied $\mathcal{L}(h|m)$ which is the procedural complexity using functions $g^{(1)}, \dots, g^{(t)}$ of a fixed size $m$. More formally , $\forall m \geq n,\ \mathcal{L}(h|m) := $ smallest $t$ such that $\exists\ g^{(1)}, \dots, g^{(t)} \in F^*(m,q)$ such that $ \pr_{[n]} \circ g^{(t)} \circ \dots \circ g^{(1)} = h \circ \pr_{[n]}$.
Here, we also use $\mathcal{L}^*(h) := \min(\{  \mathcal{L}(h|m)\ |\ n \leq m \})$ which is the procedural complexity with a size arbitrarily big.
Let $\Omega(h)$ be the number of non-trivial coordinate functions of $h$.
Theorem~\ref{th_procedural_complexity} shows that the procedural complexity of an $\AN h$ is equal to $\kappa^{min}(h)+\Omega(h)$.
Furthermore, it shows that the minimum procedural complexity is reached when we use $\kappa^{\min}(h)$ additional automata.
It is directly deduced from Lemma~\ref{lem_procedural_complexity_1} and Lemma~\ref{lem_procedural_complexity_2}.
\begin{theorem}\label{th_procedural_complexity}
	Let $h \in  F(n,q)$ and $k := \kappa^{\min}(h)$.
	We have $\mathcal{L}^*(h) = \mathcal{L}(h|n + k) = \Omega(h) + k $.
\end{theorem}

In Lemma~\ref{lem_procedural_complexity_1}, we prove that $\mathcal{L}^*(h) \leq \Omega(h) + \kappa^{min}(h)$.
We use the fact that by definition of $ k:= \kappa^{min}(h)$ there is $f \in F(n+k,q)$ and $w \in \Pi([n+k])$ such that the $n+k$ instructions $f^{w_1}, \dots, f^{w_{n+k}} \in F^*(n+k,q)$ compute $h$.
With that, we already have $\mathcal{L}(h|n + k) \leq n + k$.
Furthermore, for each $i$ such that $h_i$ is trivial, we can remove the function $f^i$ of the list of instructions and still compute $h$.
As a result, we have $\mathcal{L}(h|n + k) \leq n + k - (n - \Omega(h) ) = \Omega(h) + k$, and by definition of $\mathcal{L}^*(h)$ we have $\mathcal{L}^*(h) \leq \mathcal{L}(h|n + k)$.

\begin{lemma} ~\label{lem_procedural_complexity_1}
	Let $h \in  F(n,q)$ and $k := \kappa^{\min}(h)$.
	We have $\mathcal{L}^*(h) \leq \mathcal{L}(h|n + k) \leq \Omega(h) + k$.
\end{lemma}

In Lemma~\ref{lem_procedural_complexity_2}, we prove that $\Omega(h) + k \leq \mathcal{L}^*(h)$ with $k := \kappa^{\min}(h)$.
To do so, we take a set of functions $g^{(1)} \dots, g^{(t)} \in F^*(m,q)$ which compute $h$.
We consider an order $w \in \Pi([n])$ which updates all coordinate of $[n]$ in the same order that \flo{ $g^{(1)} \dots, g^{(t)}$ update them for the last time.} Then we prove that $h$ can be sequentialized respecting $w$ with less than $\mathcal{L}^*(h) - \Omega(h)$ additional automata. Let $J = \{ j_1, \dots, j_\ell \} $ be the set of steps such that either $g^{(j_i)}$ updates a coordinate of $]n,m]$, either it updates a coordinate of $[n]$ that will be updated again later. 
We have $\ell = \mathcal{L}^*(h) - \Omega(h)$. 
Then, we define $c: A^{n} \to A^{k}$ such that $c_i(x)$ equals $( g^{(j_i)} \circ \dots \circ g^{(1)} (x(0)^{m-n}) )_a$ with $a$ the coordinate updated by $g^{(j_i)}$. 
Then, we prove that $c$ is a proper coloring of the confusion graph $G_{h,w}$ and that $\Omega(h) + k \leq \mathcal{L}^*(h)$.
\begin{lemma} ~\label{lem_procedural_complexity_2}
	Let $h \in  F(n,q)$ and $k := \kappa^{\min}(h)$.
	We have $  \Omega(h) + k \leq \mathcal{L}^*(h) $.
\end{lemma}

In \cite{Gadouleau2011}, Proposition~12 states that $\forall h \in F(n,q)$, we have $\mathcal{L}(h|n-1) \leq 2n-1$. In Corollary~\ref{cor_sup_worst_l} bellow, we refine this bound using Theorem~\ref{th2}, Theorem~\ref{th_procedural_complexity} and the fact that $\forall h \in F(n,q), \Omega(h) \leq n$.

\begin{corollary}\label{cor_sup_worst_l}
	For all $h \in F(n,q)$, $\mathcal{L}(h|m) \leq m$ with $ m := n+ \lceil n/2 + log_q(n/2+1) \rceil$.
\end{corollary}

In the following Corollary~\ref{cor_inf_worst_l},  we give a lower bound for the procedural complexity with unlimited memory.
It is a direct corollary of Theorem~\ref{th_procedural_complexity}, Lemma~$\ref{lem_a_b}$, Theorem~$\ref{th_km_2}$, Theorem~$\ref{th_km_s}$ in which we construct an $\AN$ $h$ without trivial coordinate functions (and thus we have $\Omega(h) = n$).

\begin{corollary}\label{cor_inf_worst_l}
	For all $n,q \geq 2$ there is $h \in F(n,q)$ such that $\mathcal{L}^*(h) \geq n+ \lfloor n/3 \rfloor$.
	Furthermore, if $q \geq 4$ there is  $h \in F(n,q)$ such that $\mathcal{L}^*(h) \geq n+ \lfloor n/2-log_q(n) \rfloor$.	
\end{corollary}

\section{Bound for $\kappa^{\min}(h)$ using interaction graph} \label{section_bound_path}

Let us now present a way to upper bound $\kappa^{\min}(h)$ for an $\AN$ $h$ using the pathwidth of the interaction graph of $h$ \cite{Bodlaender1991}.

\begin{definition}[Pathwidth]
	A path decomposition of an undirected graph $G = (V,E)$ is a sequence of subsets $X_1, \dots, X_p$ of vertices such that
	\begin{itemize}
		\item $\forall (v,v') \in E,\ \exists X_i$ such that $v,v' \in X_i$.
		\item If $v \in X_i$ and $v \in X_j$ with $i < j$ then $\forall k \in [i,j], v \in X_k$ 
	\end{itemize}
	The size of a path decomposition is the size of the largest $X_\ell$ minus one.
	The \emph{pathwidth} $\Pw(G)$ is the minimum size of a path decomposition of G.
\end{definition}

Theorem~\ref{th_path} shows that the pathwidth of the graph $IG^*(h)$ is an upper bound for $\kappa^{\min}(h)$. It can be deduced directly from Lemma~\ref{lem_path_a} and Lemma~\ref{lem_path_b}.

\begin{theorem} \label{th_path}
	For any $\AN$ $h$, $\kappa^{\min}(h) \leq \Pw(\IG^*(h))$.
\end{theorem}

\flo{
Lemma~\ref{lem_path_a} shows that from a path decomposition of a graph $G$ of size $s$, we can construct a partition $c$ of its vertices in $s$ sets, and an update schedule $u$ with properties allowing an efficient sequentialization by Lemma~\ref{lem_path_b}.
We define $c$ (resp. $u$) using a greedy algorithm. We iterate the subsets $X_1, \dots X_n$ of the path decomposition and choose the value $c(i)$ (resp. $u(i)$) the first (resp. last) time we see $i$.}

\begin{lemma}\label{lem_path_a}
	Let $G = ([n],E)$ be an undirected graph and let $s = \Pw(G)$.
	Then there are functions $ c: [n] \to [s]$ and $u \in \Pi([n])$ with the following property.
	\flo{For all $i\in [n],$  we have either $1$) for all $k$ neighbor of $i$ in $G$ we have $u(i) \leq u(k)$ or $2$) for all $j,k \in [n]$ with $c(i) = c(j)$, $u(i) < u(j)$ and $k$ neighbor of $j$ in $G$ we have $u(i) \leq u(k)$.}
\end{lemma}

\flo{Lemma~\ref{lem_path_b} shows how to use $c$ and $u$ defined in Lemma~\ref{lem_path_b} to sequentialize $h$ respecting $u$.}
Each additional automaton $j$ (denoted from $1$ to $s$) computes the sum modulo $q$ of the images $\{\ h_i(x)\ |\ i \in [n] $ and $c(i) = j\ \}$. Then, each automaton of coordinate $j$ \flo{ can compute $h_j(x)$}, either because all neighbors of $j$ in $G$ have not be updated yet, or because it can compute all $h_j(x)$ such that $i \neq j$ and $ c(i) = c(j)$.

\begin{lemma} \label{lem_path_b}
	Let $h \in F(n,q)$.
	Let $G = \IG^*(h)$.
	If we have $c: [n] \to [s]$ and $u \in \Pi([n])$ such that $G,c,u$ have the same properties as in Lemma~\ref{lem_path_a}, then we have $\kappa(h,u) \leq s$.
\end{lemma}

\section{Conclusion and future research}\label{section_conclusion}

\flo{
We have seen that $\lfloor n/2 - \log_q(n)\rfloor  \leq \kappa^{\min}_{n,q} \leq \kappa_{n,q} \leq \lceil n/2 + log_q(n/2+1) \rceil$.
Thus, for any fixed $n$, the limit of $\kappa^{min}_{n,q}$ and $\kappa_{n,q}$ when $q$ tends to infinity is $n/2$.
It is an argument in favor of the conjecture made in \cite{Bridoux0PST17} which states that for any $n$ and $q$, $\kappa_{n,q} = \lfloor n/2 \rfloor$ and which is still open.
It would be interesting to investigate a variant of the problem presented in this paper, where additional automata are forbidden but several updates of the same automaton are allowed. 
The task is then to know, for given $n$ and $q$, the minimum time $t(q,n)$ such that  $\forall h \in F(n,q),\ \exists f \in F(n,q),\ w \in [n]^{t'}$ with $t' \leq t(q,n)$ such that $f^w = h$.
The value of $t(2,2)$ is not defined because for the $\AN$ $h \in F(2,2)$ such that $(0,0) \xrightarrow{h} (0,1) \xrightarrow{h} (1,1) \xrightarrow{h} (1,0) \xrightarrow{h} (0,0)$ there are no such $f$. However, with computers, we established that $t(3,2)=22$.
We can easily see that $ \mathcal{L}_{n,q} := \max(\{\mathcal{L}(h)\ |\ h \in F(n,q)\})$ is a lower bound for $t(n,q)$, and in \cite{Gadouleau2011}, it is stated that $2n-1 \leq \mathcal{L}_{n,q} \leq 4n-3$.}

	\nocite{*}
	\bibliography{sources}
	
\appendix
\setcounter{lemma}{0}
\setcounter{theorem}{0}
\newpage

\section{Proof of Theorem~\ref{th1}}

We can deduce Theorem~\ref{th1} directly from Lemma~\ref{lem1} and Lemma~\ref{lem2}.
\begin{theorem} [Theorem~\ref{th1}]
	Let us consider $h \in F(n,q)$ and the sequential update schedule $u \in \Pi([n])$. 
	Then we have $\kappa(h,u) = \lceil \log_q(\chi(G_{h,u})) \rceil$.
\end{theorem}

Lemma~\ref{lem1} shows that we can use any $(f,w)$ which sequentializes $h$ respecting $u$ to construct a proper coloring of $G_{h,u}$.
Indeed, we can color the vertices of the graph $G_{h,u}$ using the values of the additional automata of $f$ after their update.
Thus, this coloring does not use more than $q^k$ colors with $k$ the number of additional automata of $f$. 

\begin{lemma}\label{lem1}
	Let us consider $h \in F(n,q)$ and the sequential update schedule $u \in \Pi([n])$. 
	Then we have $\lceil \log_q(\chi(G_{h,u})) \rceil \leq \kappa(h,u)$.
\end{lemma}
\begin{proof}
	Without loss of generality, let us say that $u$ is the canonical sequential update schedule $(1,2,\dots,n)$.
	Let $k := \kappa(h,u)$, $m:= n+k$, $f \in F(m,q)$ and $w \in \Pi([m])$ respecting $u$ such that $\pr_{[n]} \circ f^w = h \circ \pr_{[n]}$.
	Let $x$, $x'$ be two neighbors in the confusion graph $G_{h,u}$.
	Let $y := (\letterA)^k$ (a word of size $k$ containing only the letter $\letterA$).
	Let $z := xy$ and $z' := x'y$.
	Let us prove that $f^w(z)_{[n+1,m]} \neq f^w(z')_{[n+1,m]}$.
	For the sake of contradiction, let us say that $f^w(z)_{[n+1,m]} = f^w(z')_{[n+1,m]}$.
	Since $x$ and $x'$ are neighbors in $G_{h,u}$, we know that
	$h(x) \neq h(x')$ and $\exists\ i \in [n],$ $h^{[i]}(x) =  h^{[i]}(x')$.
	Let us consider the biggest of these $i$. 
	So we have $h^{[i+1]}(x) \neq  h^{[i+1]}(x')$ and then $h_{i+1}(x) \neq h_{i+1}(x')$.
	Let $j = w(i+1)$.
	Let us prove that $f^{w_1,\dots,w_{j-1}}(z) =  f^{w_1,\dots,w_{j-1}}(z')$.
	First, we have $f^{w_1,\dots,w_{j-1}}(z)_{[n]} = h^{[i]}(x) = h^{[i]}(x') = f^{w_1,\dots,w_{j-1}}(z')_{[n]}$.
	Furthermore, for all $a \in [n+1,m]$ which is not updated before the step $j$ in $w$ we have $f^{w_1,\dots,w_{j-1}}(z)_a = y_{a-n} = f^{w_1,\dots,w_{j-1}}(z')_a$.
	Finally, for all $a \in [n+1,m]$ updated before the step $j$ in $w$ we have $f^{w_1,\dots,w_{j-1}}(z)_a = f^{w_1,\dots,w_{j-1}}(z')_a$ because we assumed that $f^w(z)_{[n+1,m]} = f^w(z')_{[n+1,m]}$.
	As a result, $f^{w_1,\dots,w_{j-1}}(z) = f^{w_1,\dots,w_{j-1}}(z')$.
	However, $f_{w_j} \circ f^{w_1,\dots,w_{j-1}}(z) = h_{i+1}(x) \neq h_{i+1}(x') =  f_{w_j} \circ f^{w_1,\dots,w_{j-1}}(z')$. 
	This is a contradiction. Consequently, we have, $f^w(z)_{[n+1,m]} \neq f^w(z')_{[n+1,m]}$.
	More generally, if $x$ and $x'$ are neighbors in $G_{h,u}$ then $f^w(xy)_{[n+1,n+k]} \neq f^w(x'y)_{[n+1,n+k]}$. In other words, $c: x \mapsto f^w(xy)_{[n+1,n+k]}$ gives a proper coloring of the confusion graph $G_{h,u}$. As a result, the confusion graph needs at most $q^k$ colors because $f^w(xy)_{[n+1,n+k]}$ is a word of size $k$ on the alphabet $q$. 
	Thus, $\chi(G_{h,u}) \leq q^k$ and $\lceil \log_q(\chi(G_{h,u})) \rceil \leq k \leq \kappa(h,u)$.
\end{proof}

Conversely, Lemma~\ref{lem2} states that we can construct a couple $(f,w)$ which sequentializes $h$ respecting $u$ from a proper coloring of $G_{h,u}$.
If this coloring uses less than $q^k$ colors then $f$ is of size at most $n+k$ and then the cost of sequentialization is at most $k$. 

\begin{lemma}\label{lem2}
	Let us consider the $\AN$ $h \in F(n,q)$ and the sequential update schedule $u \in \Pi([n])$. 
	Then we have $\kappa(h,u) \leq \lceil \log_q(\chi(G_{h,u})) \rceil$.
\end{lemma}
\begin{proof}
	Let $k := \lceil \log_q(\chi(G_{h,u})) \rceil$, $m:= n+k$.
	Let $A := [0,q[$.
	Let $w \in \Pi([m])$ which first update the $k$ last automata and then the $n$ first automata in the same order than $u$.
	In other words, $w := (n+1,\dots,n+k,u_1,\dots,u_n)$.
	Let $c:A^n \to A^k$ be a proper coloring of $G_{h,u}$.
	For all $i \in [n]$, let us define $p^{(i)}: A^m \to P(A^n)$ with $P(A^n) := \{E\ |\ E \subseteq A^n \}$ the set of subsets of $A^n$.
	First, $p^{(1)}: z \mapsto \{ z_{[n]} \}$ and then $\forall i \in [2,n]$, $p^{(i)}: z \mapsto \{ x \in A^n\ |\ h^{ \{ u_1, \dots, u_{i-1} \} }(x) = z_{[n]}$ and $c(x) = z_{[n+1,m]} \}$.
	Let $f \in F(n+k,q)$ such that:
	\begin{itemize}
		\item $\forall i \in [n+1,m],\ f_i = c_i \circ pr_{[n]}$.
		\item $\forall i \in [n],\ f_{u_i}: z \mapsto z_{u_i}$ if $p^{(i)}(z) = \emptyset$ and $h_{u_i}(x)$ with $x \in p^{(i)}(z)$ otherwise.
	\end{itemize}
	Let us prove that $\pr_{[n]} \circ f^w = h \circ \pr_{[n]}$.
	Let $x \in A^n$ and $z \in A^m$ with $z_{[n]} = x$. 
	By, induction let us prove that, 
	$$ \forall i \in [0,n],f^{w_1,\dots,w_{k+i}}(z) = h^{ \{ u_1 , \dots , u_i \} }(x) c(x).$$
	First, for $i=0,$ we have, 
	$$f^{w_1,\dots,w_{k}}(z) = f^{n+1,\dots,n+k}(z) = x (c_1(x),c_2(x),\dots,c_k(x)) = x c(x).$$
	Second, let $i \in [n]$ and let us suppose that,
	$$f^{w_1,\dots,w_{k+(i-1)}}(z) = h^{\{ u_1 , \dots , u_{i-1} \}}(x) c(x).$$
	We have $f_{w_{k+i}} \circ f^{w_1,\dots,w_{k+(i-1)}}(z) = h_{u_i}(x')$ with $x' \in p^{(i)}(f^{w_1,\dots,w_{k+(i-1)}}(z))$.
	We have $x \in p^{(i)}(f^{w_1,\dots,w_{k+(i-1)}}(z))$ because $(f^{w_1,\dots,w_{k+(i-1)}}(z))_{[n+1,m]} =c(x)$ and $(f^{w_1,\dots,w_{k+(i-1)}}(z))_{[n]}$ = $h^{ \{ u_1 , \dots , u_{i-1} \} }(x)$.
	Let us prove that $h_{u_i}(x') = h_{u_i}(x)$. 
	For the sake of contradiction let us say that $h_{u_i}(x') \neq h_{u_i}(x)$.
	Thus, $h(x) \neq h(x')$.
	However, $x,x' \in p^{(i)}(f^{w_1,\dots,w_{k+(i-1)}}(z))$ thus $h^{\{ u_1 , \dots , u_{i-1} \}}(x) = h^{\{ u_1 , \dots , u_{i-1} \}}(x')$ and $c(x) = c(x')$.
	Consequently, $x$ and $x'$ are neighbors in the confusion graph but they have the same color. This is a contradiction.
	As a result, $h_{u_i}(x') = h_{u_i}(x)$.
	Thus, $ \forall i \in [0,n],f^{w_1,\dots,w_{k+i}}(z) = h^{\{ u_1 , \dots , u_{i} \}}(x) c(x)$.
	As a consequence, $f^w(z) = h(x)\ c(x)$ and $\pr_{[n]} \circ f^w = h \circ \pr_{[n]}$.
	And since $f$ has $k$ additional automata, we have $\kappa(h,u) \leq k = \lceil \log_q(\chi(G_{h,u})) \rceil$.
\end{proof}

\section{Proof of Lemma~\ref{lem_kp_lower_bound}}

To prove Lemma~\ref{lem_kp_lower_bound}, we can construct a couple $(h,u)$ such that $G_{h,u}$ has a clique of size $q^{n/2}$.
Since the chromatic number of a graph is at least \flo{the size of its biggest clique}, we have $\chi(G_{h,u}) \geq q^{n/2}$.
As a result, $\kappa_{h,u} = \log(\chi(G_{h,u})) \geq n/2$ and we get Lemma~\ref{lem_kp_lower_bound} from that.

\begin{lemma} 
	[Lemma~\ref{lem_kp_lower_bound}]
	For all $ q \geq 2$ and $n \in \mathbb{N},$ we have $\kappa_{n,q} \geq \lfloor n/2 \rfloor$.
\end{lemma}

\begin{proof}
	Let $k := \lfloor n/2 \rfloor$.
	Let us consider $h \in F(n,q)$ such that: 
	\begin{itemize}
		\item $\forall i \in [k],\ h_i: x \mapsto x_{i+k}$
		\item $\forall i \in [k+1,2k],\ h_i: x \mapsto x_{i-k}$
		\item If $n$ is odd let $h_n: x \mapsto x_n$.
	\end{itemize}
	We also consider the canonical sequential update schedule $u := (1,2,\dots,n)$. 
	Let us consider the set of all configurations $X$ which have only $\letterA$ in their second half. 
	In other words, $X := \{ x \in A^n\ |\ x_{[k+1,n]} = (\letterA)^{n-k} \}$ ($(\letterA)^{n-k}$ beeing a word of size $n-k$ containing only the letter $\letterA$).
	Let $x, x' \in X$ such that $x \neq x'$.
	We have $x_{[k+1,n]} = (\letterA)^{n-k} = x'_{[k+1,n]}$. Thus, $x_{[k]} \neq x'_{[k]}$ and $\exists i \in [k]$ such that $x_i \neq x'_i$ and $h_{i+k}(x) = x_i \neq x'_i = h_{i+k}(x')$. Thus, $h(x) \neq h(x')$.
	However, when we update the first half of the automata, $x$ and $x'$ both become the configuration $(\letterA)^n$. 
	Indeed, $\forall i \in [k], f_i(x) = x_{i+k} = \letterA $.
	Then, we have $h^{[k]}(x) = (\letterA)^{n} = h^{[k]}(x')$.
	As a result, $(x,x')$ are neighbors in $G_{h,u}$. As a consequence, every two distinct vertices of X are neighbors.
	Thus, $X$ is a clique. Moreover, $X$ is a clique of size $q^{k}$. Thus, $\chi( G_{h,u} ) \geq q^{k}$ and $\kappa(h,u) \geq
	\lceil log_q( \chi( G_{h,u} ) ) \rceil \geq \lceil log_q(q^{k}) \rceil = k = \lfloor n/2 \rfloor $. 
	Hence, $\forall q \geq 2 , \forall n \in \mathbb{N},\ \kappa_{n,q}  \geq \lfloor n/2 \rfloor$.
\end{proof}

\remark{
	In~\cite{Gadouleau2011}, Theorem~5 shows that if $h \in F(n,q)$ is a permutation, then for any $u \in Pi([n])$ we have $\kappa(h,u) \leq n/2$ if $n$ is even and $\lfloor n/2 \rfloor +1$ otherwise. As a result, the problem is almost solved for the permutations.
}

\section{Proof of Theorem~\ref{th2}}

\begin{theorem}
	[Theorem~\ref{th2}]
	For all $ n \in \mathbb{N},\ q \geq 2$ we have $\kappa_{n,q} \leq \lceil n/2 + log_q(n/2+1) \rceil$.
\end{theorem}
\begin{proof}
	Let $h \in F(n,q)$ and $A := [0,q[$.
	Without loss of generality, let us say that $u$ is the canonical sequential update schedule $(1,2,\dots,n)$.
	Let $E$ be the set of edges of the confusion graph $G_{h,u}$.
	Let $X = \{ X_1, \dots, X_p \}$ be a partition of $A^n$, such that $x,x'$ are in the same set $X_i$ if and only if the two following conditions are respected:
	\begin{itemize}
		\item They are equal on the second half of the coordinates which will be updated in $u$.
		In other words, $x_{ \{u(n/2+1),\dots, u(n) \} } = x'_{ \{u(n/2+1),\dots, u(n) \} } $ or, more simply, $x_{ ] n/2 , n ]} = x'_{ ] n/2 , n   ]}$ because we said that $u = (1,2,\dots,n)$.
		\item They have the same image by $h$. In other words, $h(x) = h(x')$.
	\end{itemize}
	
	For all $ x \in A^n,$ let us denote by $X(x)$ the set $X_i \in X$ which contains $x$.
	Let $x^{(1)} \in X_1,\ x^{(2)} \in X_2,\ \dots,\ x^{(p)} \in X_p$.
	Let us consider the undirected graph $G' = (X,E')$ where two sets $X_i$ and $X_{i'}$ are neighbors in $G'$ if and only if there are two configurations $x \in X_i$ and $x' \in X_{i'}$ neighbors in the confusion graph $G_{h,u}$.
	Without loss of generality, let us consider the neighbors $N$ of $X_1$ in $G'$.
	If $X_j \in N$ then $\exists x \in X_1, x' \in X_j$ such that $ \exists i \in [n], h^{[i]}(x) = h^{[i]}(x')$ and $h(x) \neq h(x')$.
	Let us split $N$ in $n/2+1$ sets:
	\begin{itemize}
		\item Let us denote by $N_{[n/2]}$ the set of sets $X_j$ such that $ \exists i \in [n/2], x \in X_1, x' \in X_j$ such that $h^{[i]}(x) = h^{[i]}(x')$ and $h(x) \neq h(x')$.
		Since $h^{[i]}(x') = h^{[i]}(x),$ we have $x'_{]i,n]}= x_{]i,n]}$ and $x'_{]n/2,n]} = x_{]n/2,n]} = x^{(1)}_{]n/2,n]}$ because $i \leq n/2$.
		In other words, $\forall X_j \in N_{[n/2]},$ we have $x' \in X_j$ such that $x'_{]n/2,n]} = x^{(1)}_{]n/2,n]} $.
		However, there is only $q^{n/2}$ such configurations $x'$.
		Thus, $|N_{[n/2]}| \leq q^{n/2}$.
		\item For all $i \in [n/2+1,n]$, let us denote by $N_i$, the set of sets $X_j$ such that, $\exists x \in X_1, x' \in X_j$ such that $h^{[i]}(x) = h^{[i]}(x')$ and $h(x) \neq h(x')$.
		Let $X_j \in N_i$ and let $x \in X_1, x' \in X_j$ such that $h^{[i]}(x) = h^{[i]}(x')$.
		Thus, we have $x^{(j)}_{]i,n]} = x'_{]i,n]} = x_{]i,n]} = x^{(1)}_{]i,n]}$ because $i > n/2$.
		Thus, the value of $x^{(j)}_{[n/2+1,n]}$ is fixed on the interval $[i,n]$ and can vary only on the interval $[n/2+1,i]$. 
		As a result,the second half of $x^{(j)}$ can take $q^{i-n/2}$ values.
		Furthermore, $h_{[i]}(x^{(j)}) = h_{[i]}(x) = h_{[i]}(x') = h_{[i]}(x^{(1)})$.
		Thus, the value of $h(x^{(j)})$ is fixed on the interval $[i]$ and can vary only on the interval $[i,n]$. 
		As a result, $h(x^{(j)})$ can take $q^{n-i}$ different values.
		Now if two configurations $x'$ and $x''$ have the same image by $h$ and are equal one their second half then they are in the same set $X_j$.
		Thus, $|N_{i}| \leq q^{i-n/2}*q^{n-i} = q^{n/2}$.
	\end{itemize}
	We have $N = N_{[n/2]} \cup N_{n/2+1} \cup \dots \cup N_{n}$.
	Thus, $|N| \leq (n/2+1)q^{n/2}$.
	As a consequence, the degree of $X^1$ in $G'$ is less than $(n/2+1)q^{n/2}$ (strictly less because $X^1$ is in $N$ but is not neighbor of himself).
	As a result, $\chi(G') \leq d(G')+1 \leq (n/2+1)q^{n/2}$ with $d(G')$ the degree of $G'$.
	We can see that any coloring of this graph $G'$ gives a proper coloring of the confusion graph.
	Indeed, we can color all the configurations of a set $X_i$ in $G_{h,u}$ as we color $X_i$ in $G'$.
	If two configurations $x$ and $x'$ are neighbors in the confusion graph $G_{h,u}$, then $X(x)$ and $X(x')$ are neighbors in $G'$ and will not have the same color.
	Thus, $\chi(G_{h,u}) \leq \chi(G') \leq (n/2+1)*q^{n/2}$.
	As a consequence, according to Theorem~\ref{th1}, we have, $\kappa(h,u) \leq \lceil n/2 + \mathrm{log_q}(n/2 +1) \rceil$.
	Hence, $ \forall n \in \mathbb{N}, \kappa_{n,q} \leq \lceil n/2 + \mathrm{log_q}(n/2 +1) \rceil$.
\end{proof}

\section{Proof of Lemma~\ref{lem_a_b}}
\begin{lemma}[Lemma~\ref{lem_a_b}]
	Let $n, k \in \mathbb{N} $ and $q \geq 2$.
	If there is a function $b: \binom{[2k]}{k} \to  [0,q[^n$ such that the sets $b(E)[E]$ with $E \in \binom{[2k]}{k}$ are disjoint then there exists $h \in F(n,q)$ without trivial coordinate function, with $\kappa^{min}(h) \geq k$.
\end{lemma}
\begin{proof}
	Let $B := \bigcup\limits_{E \in \binom{[2k]}{k} } b(E)[E]$.
	Let $a: B \to [0,q[^n$ such that $\forall E \in \binom{[2k]}{k},\ \forall x \in b(E)[E],\ a(x) = E$.
	Let $h \in F(n,q)$ such that:
	$\forall x \in  B$, 
	\begin{itemize}
		\item $h_{a(x)}(x) = x_{[2k]\ \setminus \ a(x)}$.
		\item $h_{[2k]\ \setminus \ a(x)}(x) = x_{a(x)}$.
		\item $\forall i \in [2k+1,n], h_i(x) = \letterA$.
	\end{itemize} 
	and $\forall x \in [0,q[^n\ \setminus\ B,\ h(x) = (\letterA)^n$.
	We can see that $h$ does not have any trivial coordinate function.
	Indeed, for all $i \in [2k+1,n]$ we have $h_i: x \mapsto 0$ which is nontrivial. 
	Furthermore, if we take $x,y \in b(E)[E]$ with $E = [k+1,2k]$, and $x_E = (\letterA)^k$ and $y_E = (\letterB)^k$, we see that 
	$$\forall i \in [k],\ h_{i}(x) = x_{n/2+i} = \letterB \neq \letterA = y_{n/2+i} = h_i(y).$$
	However, $\forall i \in [k], i \notin E$ and thus $x_i = y_i$ because $x,y \in E$. Thus, either $h_i(x) \neq x_i$ or $h_i(y) = y_i$.
	Either way, $h_i$ is nontrivial.  
	Thus, for all $i \in [k]$, $h_i$ is nontrivial.
	The same way, we can prove that there are no trivial coordinate functions whose index is in $[k+1,2k]$.
	As a result, $h$ does not have any trivial coordinate function.
	Let us prove that $\forall u \in \Pi([n]),\ \kappa(h,u) \geq k$.
	Let us consider the sequential update schedule $u \in \Pi([n])$.
	Let $E \in \binom{[2k]}{k}$ be the set of the $k$ first automata of $[2k]$ updated in $u$. 
	Let $E' =  [2k]\ \setminus E$. Furthermore, let $i$ be the first step at which all automata of $E$ are updated in $u$.
	In other words, we have $E \subseteq \{ u_1,\dots, u_i \} $ and $E' \cap \{ u_1,\dots, u_i \} = \emptyset$.
	Let $z = b(E)$.  We will prove that $z[E]$ is a clique in the confusion graph $G_{h,u}$.
	Let $x,y \in z[E]$ with $x \neq y$.
	First let us prove that $h^{\{u_1,\dots,u_i\}}(x) = h^{\{u_1,\dots,u_i\}}(y ) $.
	We have:
	\begin{itemize}
		\item $h^{\{u_1,\dots,u_i\}}(x)_{E} = h^{\{u_1,\dots,u_i\}}(x)_{a(x)} = x_{[2k] \setminus a(x)} = x_{E'}  = z_{E'} = y_{E'} = y_{[2k] \setminus a(y)} = h^{\{u_1,\dots,u_i\}}(y)_{a(y)} = h^{\{u_1,\dots,u_i\}}(y)_{E}$.
		\item $h^{\{u_1,\dots,u_i\}}(x)_{E'} = (x)_{E'} = z_{E'} = y_{E'} = h^{\{u_1,\dots,u_i\}}(y)_{E'}$
		because $E' \cap \{ u_1,\dots, u_i \} = \emptyset$.
		\item $\forall j \in [2k+1,n],$ with $j \in  \{u_1,\dots,u_i\}$ we have $h^{\{u_1,\dots,u_i\}}(x)_{j} = h_j(x) = \letterA = h_j(y) = h^{\{u_1,\dots,u_i\}}(y)_{j}$.
		\item $\forall j \in [2k+1,n],$ with $j \not\in  \{u_1,\dots,u_i\}$ we have $h^{\{u_1,\dots,u_i\}}(x)_{j} = x_j = z_j = y_j = h^{\{u_1,\dots,u_i\}}(y)_{j}$.
	\end{itemize}
	As a result, $h^{\{u_1,\dots,u_i\}}(x) = h^{\{u_1,\dots,u_i\}}(y)$.
	Now, $x \neq y$ and $x,y \in z[E]$. 
	Thus, $x_E \neq y_E$ and $h(x)_{E'} = x_E \neq y_E = h(y)_{E'}$.
	As a result, $x$ and $y$ are neighbors in $G_{h,u}$ and then $z[E]$ is a clique. 
	Furthermore, $z[E]$ is of size $q^k$.
	Thus , $\chi(G_{h,u}) \geq q^k$.
	As a consequence, for any sequential update schedule $u$ we have $\kappa(h,u) \geq k$ and then $\kappa^{min}(h) \geq k$.
\end{proof}

\section{Proof of Theorem~\ref{th_km_2}}

\begin{theorem} [Theorem~\ref{th_km_2}]
	For all $ q \geq 2$ and $n \in \mathbb{N}$, $\kappa^{min}_{n,q} \geq \lfloor n/3 \rfloor$.
\end{theorem}

\begin{proof}
	Let $A := [0,q[$.
	\flo{In this proof, $c^i$ refers to $i$ times the composition of $c$.}
	Let $n = 3k$. (if $n=3k+1$ or $n=3k+2$ we just add one or two useless automata and the demonstration is the same).
	Let $b: \binom{[2k]}{k} \to A^n$ such that 
	$\forall E = \{e_1,e_2,\dots,e_k\} \in \binom{[2k]}{k},\ \forall x \in b(E)[E]$ we have:
	\begin{itemize}
		\item $\forall e \in \bar{E} = [2k]\ \setminus E,\ x_{e} = \letterA$ if $e+1 \in E$ and $\letterB$ otherwise.
		\item $x_{2k+1} = \letterA$ if $1 \in E$ and $\letterB$ otherwise.
		\item $\forall \ell \in [k-1],\ x_{2k+\ell+1} = \letterA$ if $j+1 $(mod $2k$)$ \in E$ and $\letterB$ otherwise with $j = e_\ell$.
	\end{itemize}
	Let $a:B \to \binom{[2k]}{k}$ be the function which decodes the subset encoded in a configuration such that $a = g \circ c^{2k} \circ h$ with:
	
	\begin{itemize}
		\item $h: x \mapsto (x,\{1\},\{ \}, 1)$ if $x_{2k+1} = \letterA$ and $(x,\{ \},\{1\}, 1)$ otherwise.
		\item $c$ such that for all $x \in B$, $I$, $\bar{I}$ subsets of $[n]$ and $e \in [2k]$:
		\begin{itemize}
			\item If $e \in \bar{I}:$
			\begin{itemize}
				\item If $x_e = \letterA$ then $c(x,I,\bar{I},e) = (x, I \cup \{e+1\},\bar{I},e+1)$.
				\item If $x_e = \letterB$ then $c(x,I,\bar{I},e) = (x, I,\bar{I} \cup \{e+1\},e+1)$.
			\end{itemize}
			\item $e \in I:$
			\begin{itemize}
				\item If $|I| = k$ then $c(x,I,\bar{I},e) = (x, I,\bar{I} \cup \{e+1\},e+1)$. 
				\item Otherwise, let $\ell = |I|$ and $b = x_{2k+\ell+1}$.
				\begin{itemize}
					\item If $b=\letterA$ then $c(x,I,\bar{I},e) = (x, I \cup \{e+1\},\bar{I},e+1)$. 
					\item If $b=\letterB$ then $c(x,I,\bar{I},e) = (x, I,\bar{I} \cup \{e+1\},e+1)$.
				\end{itemize}
				
			\end{itemize}
		\end{itemize}
		\item $g: (x,I,\bar{I},q) \mapsto I$.
	\end{itemize}
	By induction, let us prove that: 
	$$\forall i \in [2k],\ \forall x \in b(E)[E],\ c^{i-1}(h(x)) = (x, E \cap [i], \bar{E} \cap [i],i).$$
	First, for $i =1$ we have $c^{i-1}(h(x)) = c^0(h(x)) = h(x)$.
	There are $2$ cases: 
	\begin{itemize}
		\item If $1 \in E$, then $x_{2k+1} = \letterA$ because  $x_{2k+1} = \letterA$ if $1 \in E$ and $\letterB$ otherwise.
		Thus, $h(x) = (x,\{1\},\{\},1)$.
		Furthermore, $E \cap [1] = \{1\}$ and $\bar{E} \cap [1] = \{\}$.
		As a result, we have $c^{i-1}(h(x)) = (x, E \cap [1], \bar{E} \cap [1],1)$.
		\item If $1 \in \bar{E}$, then $x_{2k+1} = \letterB$ because  $x_{2k+1} = \letterA$ if $1 \in E$ and $\letterB$ otherwise.
		Thus, $h(x) = (x,\{ \},\{1\},1)$.
		Furthermore, $E \cap [1] = \{\}$ and $\bar{E} \cap [1] = \{1\}$.
		As a result, we have $c^{i-1}(h(x)) = (x, E \cap [1], \bar{E} \cap [1],1)$.
	\end{itemize} 
	Next, let us suppose that for $i \in [2k[$, we have $c^{i-1}(h(x)) = (x, E \cap [i], \bar{E} \cap [i],i)$.
	Let $I =  E \cap [i],\ \bar{I} = \bar{E} \cap [i],\ e = i$.
	Let $c^{i-1}(h(x))= (x,I,\bar{I},e)$.
	\begin{itemize}
		\item if $e \in \bar{E},$ then we have $e \in \bar{I}$. There are two cases:
		\begin{itemize}
			\item If $e+1 \in E$ then $x_e = \letterA$ because $\forall e \in \bar{E},\ x_{e} = \letterA$ if $e+1 \in E$ and $\letterB$ otherwise.
			Then $c^{i}(h(x)) = (x, I \cup \{e+1\},\bar{I},e+1)$.
			As a result,  $c^{i}(h(x)) = (x, E \cap [i+1], \bar{E} \cap [i+1],i+1)$.
			\item If $e+1 \in \bar{E}$ then $x_e = \letterB$ because $\forall e \in \bar{E},\ x_{e} = \letterA$ if $e+1 \in E$ and $\letterB$ otherwise.
			Then $c^{i}(h(x)) = (x, I,\bar{I}\cup \{e+1\},e+1)$.
			As a result,  $c^{i}(h(x)) = (x, E \cap [i+1], \bar{E} \cap [i+1],i+1)$.				 
		\end{itemize}
		\item if $e \in E$ then we have $e \in I$. There are two cases:
		\begin{itemize}
			\item If $e+1 \in E$. Then we have $|I|<k$ because $I = E \cap [i] \subseteq E$ and $|E| = k$ and $(e+1) \in E \setminus I$.
			Let $\ell=|I|$. We have $e = e_\ell$.
			We have $x_{2k+\ell+1} = \letterA$ because $\forall j \in [k-1],\ x_{2k+\ell+1} = \letterA$ if $j+1 \in E$ and $\letterB$ otherwise with $j = e_\ell$.
			Then $c^{i}(h(x)) = (x, I \cup \{e+1\} ,\bar{I},e+1)$.
			As a result,  $c^{i}(h(x)) = (x, E \cap [i+1], \bar{E} \cap [i+1],i+1)$.	
			\item $e+1 \in \bar{E}$. There are two cases: 
			\begin{itemize}
				\item If $e = e_k$. Then $|I|=k$, thus 
				$c^{i}(h(x)) = (x, I,\bar{I} \cup \{e+1\},e+1)$.
				As a result,  $c^{i}(h(x)) = (x, E \cap [i+1], \bar{E} \cap [i+1],i+1)$.	
				\item If $e = e_i$ with $i<k$. Then we have $|I|<k$. 
				Let $\ell = |I|$. We have $e = e_\ell$.
				We have $x_{2k+\ell+1} = \letterB$ because $\forall j \in [k-1],\ x_{2k+\ell+1} = \letterA$ if $j+1 \in E$ and $\letterB$ otherwise with $j = e_\ell$.
				Then $c^{i}(h(x)) = (x, I,\bar{I}\cup \{e+1\},e+1)$.
				As a result,  $c^{i}(h(x)) = (x, E \cap [i+1], \bar{E} \cap [i+1],i+1)$.	
			\end{itemize}
			
		\end{itemize}
	\end{itemize}
	By induction, we have
	$\forall x \in A^n,\ \forall i \in [2k],\ $
	$$c^{i-1}(h(x)) = (x,E^0 \cap V(i),E^1 \cap V(i),\bar{E}^0 \cap V(i),\bar{E}^1 \cap V(i), m+i+1, r(m+i+1)).$$
	In particular, we have $c^{2k}(h(x)) = (x,E,\bar{E},q)$.
	As a consequence, $a(x) = E$.
	Thus, all the sets $b(E)[E]$ with $E \in \binom{[2k]}{k}$ are disjoint.
	Using Lemma~\ref{lem_a_b}, we conclude that $\kappa^{min}_{n,q} \geq \lfloor n/3 \rfloor$.
	
\end{proof}

\section{Proof of Theorem~\ref{th_km_s}}

\begin{theorem} [Theorem~\ref{th_km_s}]
	For all $q \geq 4 , n \in \mathbb{N}$, $\kappa^{\min}_{n,q} \geq \lfloor n/2 - \log_q(n)\rfloor$.
\end{theorem}
\begin{proof}
	Let $A := [0,q[$.
	Let $n = 2k + \lceil log_q(2k) \rceil$.
	Only in this proof, to simplify the use of modulo, we index the coordinates starting from $0$ and not from $1$.
	Furthermore, each addition or subtraction is done modulo $2k$, and we will consider that if $a < b$ then $[b,a] = [a,2k[ \cup [0,b]$.
	For all $ I \subseteq [0,2k[,$ let $\triangle_I: E \mapsto |I \cap E| - |I \setminus E| $.
	Let us consider the two functions $M:\binom{[0,2k[}{k} \to [-k,k]$ and $m:\binom{[0,2k[}{k} \to [0,2k[$ such that $\forall E \in \binom{[0,2k[}{k}$,
	\begin{itemize}
		\item $M(E) := \max(\{ \triangle_{ [i] } (E) )\ |\ i \in [0,2k[ \})$.
		\item $\triangle_{ [m(E)] } (E) = M$.
	\end{itemize}
	For instance if we have $k := 4$, and $E := \{ 2,4,5,6\}$ then,
	$$\begin{array}{|c|c|c|c|c|c|c|c|c|}
	\hline
	 i	& 0 & 1 & 2 & 3 & 4 & 5 & 6 & 7 \\ 
	\hline
	\in E	& No & No & Yes & No & Yes & Yes & Yes & No \\ 
	\hline
	\triangle_{ [0,i] } (E) & -1	& -2 & -1 & -2 & -1 & 0 & 1 & 0 \\
	\hline
	\end{array} .$$
	Furthermore, $m(E) = 6$ and $M(E) = \triangle_{ [0,6] } (E) = 1$.
	For all $E \in \binom{[0;2k[}{k}$, let us denote by $E^0, E^1, \bar{E}^0, \bar{E}^1$ the subsets of $[0,2k[$ such that
	\begin{itemize}
		\item $E^0 = \{e \in E\ |\ e - 1 \in \bar{E} \}$.
		\item $E^1 = E \setminus E^0 = \{e_1,e_2, \dots, e_p\}$ with $ e_1 - m(E) -1 < e_2 - m(E) -1 < \dots < e_p -m(E) -1$.
		\item $\bar{E}^0 = \{e' \in \bar{E}\ |\ e' + 1 \in E\} $.
		\item $\bar{E}^1 = \bar{E} \setminus \bar{E}^0 = \{e'_1,e'_2, \dots, e'_p\} $ with $e'_1 -m(E) -1 < e'_2 -m(E) -1< \dots < e'_p -m(E) -1$.
	\end{itemize}
	In other words, we sort the elements of $E^1$ and $\bar{E}^1$ in the order $m+1, m+2, \dots, 2k-1, 0, 1, \dots, m$.
	If we take again our example where $k := 4$, and $E := \{ 2,4,5,6\}$ we have $E^0 = \{2,4\},\ E^1 = \{e_1=5,e_2=6\},\ \bar{E}^0 = \{1,3\}$ and $\bar{E}^1 = \{e'_1 = 7, e'_2 = 0\}$.
	Inded we have $e_1 - m(E) - 1 = 5 -6 -1 = 6 \leq e_2 - m(E) -1 = 6-6-1= 7$ and $e'_1 - m(E) -1 = 7 -6 -1 = 0 \leq e'_2 - m(E) - 1 = 0 - 6 - 1 = 1$.
	Let $v:[0,2k[ \to A^{n-2k}$ be an injective function
	and let $v^{-1}$ be the inverse function of $v$. 
	Let $b:  \binom{[0,2k[}{k} \to A^n$ such that if $x = b(E)$ then we have,
	\begin{itemize}
		\item $x_E = (\letterA)^k$.
		\item $\forall e \in \bar{E}^0,$ $x_{e} = \letterA $ if $e + 2 \in E$ and $\letterB$ otherwise.
		\item $\forall e'_j \in \bar{E}^1,\ x_{e'_j} = \letterC$ if $e_j+ 1\in E$ and $\letterD$ otherwise.
		\item $x_{[2k,n[} = v(m(E))$.
	\end{itemize}
	Again, with the same example, for all $y \in b(E)[E]$ we have:
	$$\begin{array}{|c|c|c|c|c|c|c|c|c|}
	\hline
	i	& 0 & 1 & 2 & 3 & 4 & 5 & 6 & 7 \\ 
	\hline
	\in E	& No & No & Yes & No & Yes & Yes & Yes & No \\ 
	\hline
	y_i & \letterD	& \letterB & y_2 & \letterA & y_4 & y_5 & y_6 & \letterC \\
	\hline
	\end{array} .$$
	Indeed, 
	\begin{itemize}
		\item $y_3 =\letterA$ because $3 \in \bar{E}^0$ and $3+2 \in E$.
		\item $y_1 = \letterB$ because $1 \in \bar{E}^0$ and $1+2 \notin E$.
		\item $y_7 = \letterC$ because $e'_1 = 7 \in \bar{E}^1$ and $e_1 + 1 = 5+1 \in E$.
		\item $y_0 = \letterD$ because $e'_2 = 0 \in \bar{E}^1$ and $e_2+1 = 6+1 \notin E$.
 	\end{itemize}
	
	Let $ B =: \{ x \in b(E)[E]\ |\ E \in \binom{[2k]}{k} \}$ be the set of configuration which encodes a set $E$.
	Let us consider the function $a: B \to \binom{[2k]}{k}$ which decodes the set encoded by any configuration of $B$ such that $a = g \circ c^{2k} \circ h$ with:
	\begin{itemize}
		\item 
		$h:x \mapsto (x,\emptyset,\emptyset,\emptyset,\emptyset, m+1,0)$ with $m = v^{-1}(x_{[2k,n[})$.
		\item $c$ such that for all $I^0,I^1,\bar{I}^0,\bar{I}^1$ subsets of $[0,2k[$, $e \in [0,2k[$ and $q \in [0,3[$,
		\begin{itemize}
			\item if $q=\letterA$:
			\begin{itemize}
				\item if $x_e = \letterA$ or $\letterB$ then $c(x,I^0,I^1,\bar{I}^0,\bar{I}^1,e,q) = (x,I^0,I^1,\bar{I}^0 \cup \{e\},\bar{I}^1 ,e+1,\letterB)$.
				\item if $x_e = \letterC$ or $\letterD$ then $c(x,I^0,I^1,\bar{I}^0,\bar{I}^1,e,q) = (x,I^0,I^1,\bar{I}^0 ,\bar{I}^1\cup \{e\},e+1,\letterA)$.
			\end{itemize}
			\item if $q=\letterB$:
			\begin{itemize}
				\item if $x_{e-1} = \letterA$ then  $c(x,I^0,I^1,\bar{I}^0,\bar{I}^1,e,q) = (x,I^0 \cup \{e\},I^1,\bar{I}^0,\bar{I}^1,e+1,\letterC)$.
				\item if $x_{e-1} = \letterB$ then  $c(x,I^0,I^1,\bar{I}^0,\bar{I}^1,e,q) = (x,I^0 \cup \{e\},I^1,\bar{I}^0,\bar{I}^1,e+1,\letterA)$.
			\end{itemize}
			\item if $q=\letterC$, let $j=|I^1|$, $e' = \bar{I}^1_j$ (the $j$-th element of $\bar{I}^1$ when we sort them it in the order $m+1,m+2,\dots,2k-1,0,1, \dots ,m$ ).
			\begin{itemize}
				\item if $x_{e'} = \letterC$ then $c(x,I^0,I^1,\bar{I}^0,\bar{I}^1,e,q) = (x,I^0 ,I^1 \cup \{e\},\bar{I}^0,\bar{I}^1,e+1,\letterC)$.
				\item if $x_{e'} = \letterD$ then $c(x,I^0,I^1,\bar{I}^0,\bar{I}^1,e,q) = (x,I^0 ,I^1 \cup \{e\},\bar{I}^0,\bar{I}^1,e+1,\letterA)$.
			\end{itemize}
		\end{itemize}
		\item $g: (x,I^0,I^1,\bar{I}^0,\bar{I}^1,e,q) \mapsto I^0 \cup I^1 $.
	\end{itemize}
	With the same example, let $y \in b(E)[E]$ and let us compute $a(y)$.
	\begin{gather*}
		y = (3,1,y_2,0,y_4,y_5,y_6,2)v(m(E)) \\
		\xrightarrow{h}  (y, \emptyset, \emptyset, \emptyset, \emptyset, 7, 0)\\
		\xrightarrow{c}  (y, \emptyset, \emptyset, \emptyset, \{7\}, 0, 0)\\
		\xrightarrow{c}  (y, \emptyset, \emptyset, \emptyset, \{7,0\}, 1, 0)\\		
		\xrightarrow{c}  (y, \emptyset, \emptyset, \{1\}, \{7,0\}, 2, 1)\\
		\xrightarrow{c}  (y, \{2\}, \emptyset, \{1\}, \{7,0\}, 3, 0)\\
		\xrightarrow{c}  (y, \{2\}, \emptyset, \{1,3\}, \{7,0\}, 4, 1)\\
		\xrightarrow{c}  (y, \{2,4\}, \emptyset, \{1,3\}, \{7,0\}, 5, 2)\\
		\xrightarrow{c}  (y, \{2,4\}, \{5\} \{1,3\}, \{7,0\}, 6, 2)\\
		\xrightarrow{c}  (y, \{2,4\}, \{5,6\} \{1,3\}, \{7,0\}, 7, 0)\\
		\xrightarrow{g}  \{2,4,5,6\} = E.\\
	\end{gather*}
	Thus, we have $\forall y \in b(E)[E]$, $a(y) = E$.
	If we can prove that for all $E \in \binom{[0,2k[}{k}$ and for all $y \in b(E)[E]$, we have $a(y) = E$, then we prove that the sets $b(E)[E]$ are disjoint. Furthermore, with Lemma~\ref{lem_a_b}, we can conclude that $\kappa^{min}_{2k+log(k),q} \geq k$.
	In the remaining of the proof, we prove it formally for all $k$ and $E$.
	Let $k \in \mathbb{N}$, $E \in \binom{[0,2k[}{k}$.
	Let $r_E: \ell \mapsto \begin{cases}
	\letterA &\text{if } \ell \in \bar{E} \\
	\letterB &\text{if } \ell \in E^0\\
	\letterC &\text{otherwise }
	\end{cases}$.
	By induction, let us prove that for all $i \in [0,2k]$,
	$$c^i(h(x)) = (x,E^0 \cap V(i),E^1 \cap V(i),\bar{E}^0 \cap V(i),\bar{E}^1 \cap V(i), m+i+1, r(m+i+1)).$$
	with $V(0) = \emptyset$, and $\forall i \in [2k], V(i) = [m+1,m+i]$.
	First, let us prove that $m+1 \notin E$.
	For the sake of contradiction, let us say that $m+1 \in E$.
	Then we have $\vartriangle_{[m+1]}(E) = \vartriangle_{[m]}(E) + \vartriangle_{[m+1,m+1]}(E) = M(E) +1.$
	This is absurd because $M(E) = \max(\{ \triangle_{ [i] } (E) )\ |\ i \in [2k] \})$.
	As a result $m+1 \notin E$.
	Thus, $r_E(m+0+1)=\letterA$.
	Furthermore, 	\begin{gather*}
c^0(h(x)) = h(x) = (x,\{\},\{\},\{\},\{\},m+0+1,0) \\
= (x,E^0 \cap V(0),E^1 \cap V(0),\bar{E}^0 \cap V(0),\bar{E}^1 \cap V(0), m+0+1, r_E(m+0+1) ).
	\end{gather*}

	Next, let us suppose that the induction hypothesis hold for $i \in [0,2k[$.
	Let $I = E \cap V(i)$,
	$\bar{I} = \bar{E} \cap V(i)$, 
	$I^0 = E^0 \cap V(i)$, 
	$I^1 = E^1 \cap V(i)$, 
	$\bar{I}^0 = \bar{E}^0 \cap V(i)$, 
	$\bar{I}^1 = \bar{E}^1 \cap V(i)$, 
	$e = m+i+1$ and $q=r_E(e)$.
	Thus, $c^i(h(x)) = (x,I^0,I^1,\bar{I}^0,\bar{I}^1,e,q)$.
	There are four cases:
	\begin{itemize}
		\item $e = m+i+1 \in \bar{E}^0$.
		As a consequence, we have $q = r_E(e) = \letterA$.
		Furthermore, we have $x_{e} = \letterA$ or $\letterB$ because $\forall e \in \bar{E}^0,\ x_{e} = \letterA$ if $e+2 \in E$ and $\letterB$ otherwise.
		Thus, $c^{i+1}(h(x)) = (x,I^0,I^1,\bar{I}^0 \cup \{e\} ,\bar{I}^1,e+1,1)$.
		By definition of $\bar{E}^0$, we have $e+1 \in E^0$ and then $r_E(e+1) = \letterB$.
		Thus, $c^{i+1}(h(x))= (x,E^0 \cap V(i+1),E^1 \cap (V(i+1)),\bar{E}^0 \cap V(i+1),\bar{E}^1 \cap V(i+1), m+i+2, r_E(m+i+2))$.
		\item $e = m+i+1 \in \bar{E}^1$.
		As a consequence, we have $q = r_E(e) = \letterA$. 
		Furthermore, we have $x_{e} = \letterC$ or $\letterD$ because $\forall e'_i \in \bar{E}^1,$ $x_{e'_i} = \letterC $ if $e_i + 1 \in E$ and $\letterD$ otherwise. 
		Thus, $c^{i+1}(h(x)) = (x,I^0,I^1,\bar{I}^0,\bar{I}^1 \cup \{e\},e+1,1)$.
		By definition of $\bar{E}^1$, we have $e+1 \in \bar{E}$ and then $r_E(e+1) = \letterA$.
		Thus, $c^{i+1}(h(x)) = (x,E^0 \cap V(i+1),E^1 \cap V(i+1),\bar{E}^0 \cap V(i+1),\bar{E}^1 \cap V(i+1), m+i+2, r_E(m+i+2))$.
				
		\item $e = m+i+1 \in E^0$.
		By induction hypothesis, we have $q = r_E(e) = \letterB$.
		Furthermore, by definition of $E^0$, $e-1 \in \bar{E}^0$.
		There are two subcases:
		\begin{itemize}
			\item $e+1 \in E$.
			We have $x_{e-1} = \letterA$ because $\forall (e-1) \in \bar{E}^0,\ x_{e-1} = \letterA$ if $(e-1)+2 = e+1 \in E$ and $\letterB$ otherwise.
			Thus, $c^{i+1}(h(x)) = (x,I^0 \cup \{e\} ,I^1,\bar{I}^0,\bar{I}^1,e+1,\letterC)$.
			And since $e+1 \in E$ and $e \notin \bar{E}$, then $e+1 \in E^1$ and $r_E(e+1) = 2$.
			As a result, $c^{i+1}(h(x)) = (x,E^0 \cap V(i+1),E^1 \cap V(i+1),\bar{E}^0 \cap V(i+1),\bar{E}^1 \cap V(i+1), m+i+2, r_E(m+i+2))$.
			\item $e+1 \in \bar{E}$.
			We have $x_{e-1} = \letterB$ because $\forall (e-1) \in \bar{E}^0,\ x_{e-1} = \letterA$ if $(e-1)+2 = e+1 \in E$ and $\letterB$ otherwise.
			Thus, $c^{i+1}(h(x)) = (x,I^0 \cup \{e\},I^1,\bar{I}^0,\bar{I}^1,e+1,\letterA)$.
			And since $e+1 \in \bar{E}$, $r_E(e) = \letterA$.
			As a result, $c^{i+1}(h(x)) = (x,E^0 \cap V(i+1),E^1 \cap V(i+1),\bar{E}^0 \cap V(i+1),\bar{E}^1 \cap V(i+1), m+i+2, r_E(m+i+2))$.
			
		\end{itemize}
		\item $e = m+i+1 \in E^1$.
		As a consequence, we have $q = r_E(e) = \letterC$.
		Let $j = |I^1|$.
		Let us prove that $|\bar{I}^1| < |I^1|$.
		First, we have $|I^0| = |\bar{I}^0|$.
		Indeed, for all $u \in \bar{I}^0$, we have also $u \in \bar{E}^0$ and then $u+1 \in E^0$.
		Furthermore, $e \in E^1$ and thus $e-1 = m+i \notin \bar{E}^1$.
		Thus, $u \in [m+1,m+i+1[$, $u+1 \in [m+1,m+i+1]$. Consequently, $u+1 \in V(i)$ and $u+1 \in I^0$.
		As a result, for all $u \in \bar{I}^0,$ we have $u+1 \in I^0$.
		As a consequence, $|\bar{I}^0| \leq |I^0|$.
		Reversely, for all $v \in I^0$, $v \in E^0$ and then $v-1 \in \bar{E}^0$.
		Furthermore, $m+1 \in \bar{E}$.
		Thus, $v \in ]m+1,m+i+1]$ and $v-1 \in [m+1,m+i+1]$. As a result, $v-1 \in V(i)$ and $v-1 \in \bar{I}^0$.
		Consequently, for all $v \in I^0,$ we have $v-1 \in \bar{I}^0$.
		As a consequence, $|I^0| \leq |\bar{I}^0|$ and then $|I^0| = |\bar{I}^0|$.
		Now, $\vartriangle_{[m+i+1]}(E)$ = $\vartriangle_{[m]}(E)$ + $\vartriangle_{[m+1,m+i]}(E)$ + $|\{e\}| = M(E) + |I| - |\bar{I}| + 1 = M(E) + |I^0| + |I^1| - |\bar{I}^0| - |\bar{I}^1| + 1 = M(E)+ |I^1| - |\bar{I}^1| + 1$.
		If $|I^1| \geq |\bar{I^1}|$ then $\vartriangle_{[m+i+1]}(E) > M(E)$ which is absurd.
		Thus, $|I^1| < |\bar{I^1}|$.
		Let $e' = |\bar{I}^1_{j}|$.
		We have $e = e_{j}$ and $e'= e'_{j}$.
		There are two cases:
		\begin{itemize}
			\item $e+1 \in E$.
			Then $x_{e'} = \letterC$ because $\forall e'_j \in \bar{E}^1,$ $x_{e'_j} = \letterC $ if $e_j + 1 \in E$ and $\letterD$ otherwise.
			Thus, $c^{i+1}(h(x)) = (x,I^0 ,I^1 \cup \{e\},\bar{I}^0,\bar{I}^1,e+1,\letterC)$.
			Furthermore, $e+1 \in E$ and $e \notin \bar{E}$ then $e+1 \in E^1$ and $r_E(e+1)=\letterC$.
			As a result, $c^{i+1}(h(x)) = (x,E^0 \cap V(i+1),E^1 \cap V(i+1),\bar{E}^0 \cap V(i+1),\bar{E}^1 \cap V(i+1), m+i+2, r_E(m+i+2))$.
			\item $e+1 \in \bar{E}$.
			Then $x_{e'} = \letterD$ because $\forall e'_j \in \bar{E}^1,$ $x_{e'_j} = \letterC $ if $e_j + 1 \in E$ and $\letterD$ otherwise.
			Thus, $c^{i+1}(h(x)) = (x,I^0 ,I^1 \cup \{e\},\bar{I}^0,\bar{I}^1,e+1,0)$.
			Furthermore, $e+1 \in E$ and $e \notin \bar{E}$ then $e+1 \in E^1$ and $r_E(e+1)=\letterA$.
			As a result, $c^{i+1}(h(x)) = (x,E^0 \cap V(i+1),E^1 \cap V(i+1),\bar{E}^0 \cap V(i+1),\bar{E}^1 \cap V(i+1), m+i+2, r_E(m+i+2))$.
		\end{itemize}
	\end{itemize}
	By induction, we can see that $\forall i \in [2k[,\ c^{i+1}(h(x)) = (x,E^0 \cap V(i+1),E^1 \cap V(i+1),\bar{E}^0 \cap V(i+1),\bar{E}^1 \cap V(i+1), m+i+2, r_E(m+i+2))$.
	As a result, $a(x) = g(c^{2k}(h(x))) = g(x,E^0,E^1,\bar{E}^0,\bar{E}^1,m+2k,r_E(m+2k+1)) = E^0 \cup E^1 = E$.
	Since there is a function $a$ such that $\forall E \in \binom{[0,2k[}{k},$ $\forall x \in b(E)[E]$, $a(x) = E$, we know that all the sets $b(E)[E]$ are disjoint.
	Using Lemma~\ref{lem_a_b}, we conclude that $\kappa^{min}_{2k+log(k),q} \geq k$.

\end{proof}

\section{Proof of Lemma~\ref{lem_procedural_complexity_1}}

\begin{lemma} [Lemma~\ref{lem_procedural_complexity_1}]
	Let $h \in  F(n,q)$ and $k := \kappa^{\min}(h)$.
	We have $  \Omega(h) + k \leq \mathcal{L}^*(h) $.
\end{lemma}
\begin{proof}
	Let $h \in  F(n,q)$, $k := \kappa^{\min}(h)$, $A := [0,q[$ and $m := n+k$.
	By definition of $\kappa^{min}(h)$ there exists $f \in F(m,q)$ and $w \in \Pi([m])$ such that $f^w$ simulates $h$.
	Thus, $\pr_{[n]} \circ f^{w_n} \circ \dots f^{w_1} = h \circ \pr_{[n]}$.
	By definition, $\forall i \in [m]$, $f^{i}$ does not update more than one coordinate.
	Then, $f^{w_1}, \dots, f^{w_{m}} \in F^*(m,q)$.
	Let us consider the set $T$ of the coordinates of the trivial functions of $h$ and let $w' \in \Pi( [m] \setminus T)$ be an order respecting $w$ which does update the coordinates of $T$. In other words, $\forall i,j \in \Pi( [m] \setminus T),$ if $w(i) < w(j)$ then $w'(i) < w'(j)$.
	Let us prove that $f^{w'} = f^w$. 
	Let $h_i$ be a trivial coordinate function. Thus, $\forall x \in A^n, h_i(x) = x_i$.
	And for all $y \in A^k$ and $z:= xy$, we have $(f^w(z))_i = h_i(x) = x_i = z_i$.
	Furthermore, since $w \in \Pi([m])$, the coordinate $i$ is updated only one time in $w$ in step $j := w(i)$.
	Thus, $f_{w_j} \circ f^{w_1,\dots, w_{j-1}}(z) = (f^w(z))_i = z_i$.
	Furthermore, since $i$ is not updated before the step $j$, we have $(f^{w_1,\dots, w_{j-1}}(z))_i = z_i$.
	As a result, $f^{w_j} \circ f^{w_1,\dots, w_{j-1}} = f^{w_1,\dots, w_{j-1}}$, $f^w = f^{w_1,\dots,w_{j-1},w_j, \dots, w_m}$.
	Using the same method for all $j$ such that $h_{w_j}$ is trivial we get $f^w = f^{w'}$.
	The order $w'$ is of size $\Omega(h) + k$.
	As a result, we have $\mathcal{L}(h|n + k) \leq \Omega(h) + k $.
	And by definition of $\mathcal{L}^*(h)$ we have $\mathcal{L}^*(h) \leq \mathcal{L}(h|n + k)$.
\end{proof}

\section{Proof of Label~\ref{lem_procedural_complexity_2}}

\begin{lemma} [Label~\ref{lem_procedural_complexity_2}]
	Let $h \in  F(n,q)$ and $k := \kappa^{\min}(h)$.
	We have $  \Omega(h) + k \leq \mathcal{L}^*(h) $.
\end{lemma}

\begin{proof}
	Let $\ell  :=  \mathcal{L}^*(h),$ $m \leq n$ and $g^{(1)},\dots, g^{(\ell)} \in F^*(m,q)$ such that $ \pr_{[n]} \circ g^{(\ell)} \circ \dots \circ g^{(1)} =  h \circ \pr_{[n]}$.
	We can assume that for all $i \in [\ell]$, the function $g^{(i)}$ updates one coordinate. Otherwise, $g^{i}$ would be the identity function and we could remove it and have $\ell > \mathcal{L}^*(h)$ which is absurd.
	Let $u \in [m]^t$ such that, for all $i \in [\ell]$, $u_i$ is the coordinate updated by $g^{(i)}$.
	Let $I=\{i_1,i_2,\dots,i_p \}$ with $i_1<i_2<\dots<i_p$ the set of steps where a coordinate of $[n]$ is updated for the last time in $u$. In other words, $\forall j \in [p],\ u_{i_j} \in [n]$ and $\forall i \in [i_j+1,\ell], u_{i} \neq u_{i_j}$. 
	We know that $\Omega(h) \leq p$ because, to compute $h$, each coordinate of a nontrivial function of $h$ needs to be updated at least once.
	Indeed, if $h_i$ is nontrivial, then $\exists x \in A^n, h_i(x) \neq x_i$. 	
	If $i$ is not updated in $u$, then $\forall y \in A^{m-n},\ \pr_{i} \circ g^{(\ell)} \circ \dots \circ g^{(1)}(xy) = x_i \neq h_i(x)$ and $h$ is not computed.
	Let $k := \ell - p \leq \mathcal{L}^*(h) - \Omega(h)$.
	Let $v$ be an oder which updates all the coordinates of $[n]$ not updated by $g^{(1)}, \dots, g^{\ell}$.
	Let $u' := (u_{i_1},u_{i_2},\dots,u_{i_p})$ be an order which updates the coordinates of $[n]$ updated by $g^{(1)}, \dots, g^{\ell}$ in the same order that $u$ updates them for the last time.
	Let $w:= u v \in \Pi([n])$ be a permutation of $[n]$.
	Let $J=\{j_1,j_2,\dots,j_{k}\} = [\ell] \setminus I$ with $j_1 < j_2  < \dots < j_{k}$.
	For all $i \in [n]$, let $\tilde{g}^{(i)}: A^m \to A$ be the function which return the value of the coordinate updated by $g^{(i)}$. In other words, $\tilde{g}^{(i)} = pr_{u_i} \circ g^{(i)}$.
	Let $y := (\letterA)^{m-n}$ (a word of size $m-n$ containing only the letter $\letterA$).
	Let $c:A^n \to A^{k},$ such that, $\forall x \in A^n$, $\forall i \in [k], c_i(x) = \tilde{g}^{(j_i)} \circ g^{(j_i-1)} \circ ... \circ g^{(1)}(xy)$.
	Let us prove that $c$ give a proper coloring of the confusion graph $G_{h,w}$.
	Let $x, x' \in A^n$,be neighbors in the confusion graph $G_{h,w}$.
	In other words, $h(x) \neq h(x')$ but $\exists i \in [n], h^{  \{ w_1, w_2,\dots, w_{i} \} }(x) = h^{  \{ w_1, w_2,\dots, w_{i} \} }(x')$.
	For the sake of contradiction, let us say that $c(x) = c(x')$.
	Let $z := xy$ and $z' := x'y$.
	Let $e = \max(\{i \in [n]\ |\ h^{  \{ w_1,\dots, w_{i} \} }(x) = h^{  \{ w_1,\dots, w_{i} \} }(x') \})$.
	Let $b \in [p]$ be the last step of $u$ in which the coordinate $w_e$ is updated.
	Let $r := g^{(b)} \circ \dots \circ g^{(1)} (z)$ and $r' := g^{(b)} \circ \dots \circ g^{(1)} (z')$.
	Let us prove that $r = r'$.
	For all $a \in [m]$ not yet updated in $u$ at step $b$:
	\begin{itemize}
		\item if $a \in [n]$ then $r_a = x_a = x'_a = r'_a$ because $h^{  \{ w_1, w_2,\dots, w_{e} \} }(x) = h^{  \{ w_1, w_2,\dots, w_{e} \} }(x')$ and thus $x_{ [n] \setminus \{ w_1, \dots, w_e \} } = x'_{ [n] \setminus \{ w_1, \dots, w_e \} }$.
		\item if $a \in [n+1,m]$ then $r_a = z_a = y_{a-n} = z'_a = r'_a$.
	\end{itemize}
	For all $a \in [m]$ already updated in $u$ at step $b$:
	\begin{itemize}
		\item if $a \in [n]$ and $a$ is updated for the last time then $r_a = h_a(x) = h_a(x') = r'_a$ because $h^{  \{ w_1, w_2,\dots, w_{i} \} }(x) = h^{  \{ w_1, w_2,\dots, w_{i} \} }(x')$ and thus $h(x)_{ [n] \setminus \{ w_1, \dots, w_e \} } = h(x')_{ [n] \setminus \{ w_1, \dots, w_e \} }$.
		\item Otherwise, let $d < b$ be the last step in $u$ before $b$ such that $a$ is updated.
		In other words, $u_{d} = a$ and $\forall i \in ] d, b [, u_i \neq a$.
		We have $r_a = \tilde{g}^{(d)} \circ g^{(d-1)} \circ \dots \circ g^{(1)}(z) = c_{d}(x) = c_{d'}(x) = \tilde{g}^{(d)} \circ g^{(d-1)} \circ \dots \circ g^{(1)}(z') = r'_a$.
	\end{itemize}
	Thus, we have $g^{(b)} \circ \dots \circ g^{(1)} (z) = g^{(b)} \circ \dots \circ g^{(1)} (z')$ and thus $g^{(\ell)} \circ \dots \circ g^{(1)} (z) = g^{(\ell)} \circ \dots \circ g^{(1)} (z')$.
	However, $pr_{[n]} \circ g^{(\ell)} \circ \dots \circ g^{(1)} (z) = h(x) \neq h(x') \neq g^{(\ell)} \circ \dots \circ g^{(1)} (z')$.
	This is absurd, so if two configurations $x,x'$ are neighbors in the confusion graph then $c(x) \neq c(x')$.
	Thus, $c$ gives a proper coloring of the confusion graph $G_{h,w}$ and it uses at most $q^{k} = q^{\ell - p} \leq q^{\mathcal{L}^*(h) - \Omega(h)}$ colors.
	As a result, $\kappa(h,w) \leq \mathcal{L}^*(h) - \Omega(h)$.
\end{proof}

\section{Proof of Lemma~\ref{lem_path_a}}

\begin{lemma}[  Lemma~\ref{lem_path_a} ]
	Let $G = ([n],E)$ be an undirected graph and let $s = \Pw(G)$.
	Then there are functions $ c: [n] \to [s]$ and $u \in \Pi([n])$ with the following property.
	\flo{For all $i\in [n],$  we have either $1$) for all $k$ neighbor of $i$ in $G$ we have $u(i) \leq u(k)$ or $2$) for all $j,k \in [n]$ with $c(i) = c(j)$, $u(i) < u(j)$ and $k$ neighbor of $j$ in $G$ we have $u(i) \leq u(k)$.}
\end{lemma}
\begin{proof}
	Let $G = ([n],E)$, $s = \Pw(G) $ and $X_1, \dots, X_{p}$ a minimal path decomposition of $G$.
	In other words:
	\begin{itemize}
		\item $\forall i \in [n], \forall a,b \in [p]$ with $a<b$, 
		if $ i \in X_a$ and $ i \in X_{b}$ then $\forall \ell \in [a,b]$, $i \in X_\ell$.
		\item $\forall (i,j) \in E, \exists a \in [p]$ such that $i \in X_a$ and $j \in X_a$.
		\item $ \forall i \in [n], \exists a \in [p]$ such that $i \in X_a$.
		\item $\forall i \in [p],\ |X_i| \leq s+1$.
	\end{itemize}
	For all $i \in [n],$ let $X(i) = \{ X \in \{X_1,\dots,X_p\}\ |\ i \in X \}$.
	Let $b: i \mapsto \min(\{ j\ |\ X_j \in X(i) \})$ and $e: i \mapsto \max(\{ j\ |\ X_j \in X(i) \})$.
	We will assume that $\forall \{a,b\} \subseteq [p],$ we do not have $X_a \subseteq X_b$ 
	since otherwise we could remove $X_a$ and still have a valid path decomposition of same size.
	As a result, for all $a \in [p],$ there exists $j \in X_a$ such that $e(j) =a$. Indeed, if that was not the case, we would have $X_a \subseteq X_{a+1}$.
	Let $u \in \Pi([n])$ be an order respecting $e$ and $v \in \Pi([n])$ be an order respecting $b$. In other words, for all $\{i,j\} \subseteq [n],$ if $e(i) < e(j)$ then $u(i) < u(j)$ and if $b(i) < b(j)$ then $v(i) < v(j)$.
	For all $j$ in $[n]$ taken in the order $v$, let us define $c(j)$ as such:
	\begin{itemize}
		\item If, in the set of images by $c$ of $X_{b(j)}$ already defined, there are value of $[s]$ not used then let $c(j)$ be the minimal of them.
		More formally, if $\{ c(k)\ |\ k \in X_{b(j)}$ and $v(k) < v(j) \} \neq [s]$ then let $c(j) := \min(\ [s] \setminus \{ c(k)\ |\ k \in X_{b(j)}$ and $v(k) < v(j) \}\ )$.
		\item Otherwise, if there is $k \in X_{b(j)}$ such that $v(k) < v(j)$ and $e(k) = b(j)$, then let us consider the $i$ which minimize 
		$u(i)$.
		In other words, let us consider $i$ such that $u(i) = min(\{ u(k)\ |\ k \in X_{b(j)}\} )$ and let $c(j) := c(i)$.
		We remark that since $\forall k \in X_{b(j)}, b(j) \leq e(k)$, we have $e(i) = b(j)$ (and not $e(i) < b(j)$).
		\item Otherwise, let $c(j) := \letterA$.
		In this case, we have $b(j) = e(j)$ because $\forall k \in X_{b(j)} \setminus \{j\},\ b(j) < e(k)$ and by hypothesis $\forall a \in [p],$ there exists $j \in X_a$ such that $e(j) = a$.
	\end{itemize}
	We remark that with this construction of $c$, $\forall a \in [p]$ there is at most one $\{i,j\} \subseteq X_a$ such that $c(i) = c(j)$ because $|X_a| \leq s+1$.
	Let $i \in [n]$.
	First, let us consider the case where $c(i)$ is defined using the third case. 
	Then, we have $b(i) = e(i)$ and thus $X(i) = \{X_{b(i)}\}$.
	By definition of a path decomposition, for all neighbor $k$ of $i$ in $G$, we have $k \in X_{b(i)}$.
	Furthermore, $\forall k \in X_{b(i)},\  e(i) = b(i)  < e(k)$.
	Thus, $\forall k \in X_{b(i)},\  u(i) < u(k)$.
	As a result, $i$ respects the condition $1$) for all $k$ neighbor of $i$ in $G$ we have $u(i) \leq u(k)$.
	Next, let us assume that $c(i)$ is not defined using the third case. 
	Let $j \in [n]$ such that $c(i) = c(j)$ and $u(i) < u(j)$.
	Let us prove that for all $k$ neighbor of $j$ in $G$, $u(i) \leq u(k)$.
	First let us prove that we have $e(i) \leq b(j)$.
	For the sake of contradiction, let us say that $b(j) < e(i)$.
	There are $2$ cases:
	\begin{itemize}
		\item $v(i) < v(j)$. Thus, $b(i) \leq b(j) < e(i)$.
		However,
		\begin{itemize}
			\item Since, $b(i) \leq b(j) < e(i)$, we have $b(j) \in [b(i),e(i)]$ and thus $i \in X_{b(j)}$. 
			Thus, there exists $k \in X_{b(j)}$ (k := i) such that $c(k) = c(j)$ and $v(k) < v(j)$.
			As a result, $c(j)$ cannot have been defined using the first case of the definition.
			\item  We have $b(j) < e(i)$.
			Furthermore, by hypothesis, we know that there is $k \in X_{b(j)}$, such that $e(k) = b(j)$. Thus, $e(k) < e(i)$ and then $u(k) < u(i)$.
			As a consequence, we have $c(i) = c(j)$ but $u(i) \neq min(\{ u(k) | k \in X_{b(j)}\} )$.
			As a result, $c(j)$ cannot have been defined using the second case of the definition.
			\item We have $b(j) < e(i) \leq e(j)$ and thus $b(j) \neq e(j)$.
			As a result, $c(j)$ cannot have been defined using the third case of the definition.
		\end{itemize}
		\item $v(j) < v(i)$ Thus, $b(j) \leq b(i) < e(i)$.
		However,
		\begin{itemize}
			\item Since, $b(j) \leq b(i) < e(i) \leq e(j)$, we have $b(i) \in [b(j),e(j)]$ and thus $j \in X_{b(i)}$. 
			Thus, there is $k \in X_{b(i)}$ such that $c(k) = c(i)$ and $v(k) < v(i)$.
			As a result, $c(i)$ cannot have been defined using the first case of the definition. 
			\item  We have $b(i) \leq b(j) < e(i) \leq e(j)$ and thus $b(i) < e(j)$.
			Furthermore, we know that there is $k \in X_{b(i)}$, such that $e(k) = b(j)$. Thus, $e(k) < e(j)$ and then $u(k) < u(j)$.
			As a consequence, we have $c(j) = c(i)$ but $u(j) \neq min(\{ u(k) | k \in X_{b(i)}\} )$.
			As a result, $c(i)$ cannot have been defined using the second case of the definition.
			\item By hypothesis, $c(i)$ is not defined using the third case of the definition. 
		\end{itemize}
	\end{itemize}
	
	All cases raise a contradiction. As a result, we have $e(i) \leq b(j)$.
	There are $2$ cases: 
	\begin{itemize}
		\item If $e(i) < b(j)$, then let us prove that for each $k$ neighbor of $j$ in $G$, $u(i) < u(k)$ (and then $u(i) \leq u(k)$ ).
		Let $k$ be a neighbor of $j$ in $G$.
		Thus, $\exists a \in [p], k \in X_a$ and $j \in X_a$ and thus $b(j) \leq b(k)$. 
		As a consequence, $e(i) < b(j) \leq b(k) \leq e(k)$ and then $u(i) < u(k)$. 
		\item If $e(i) = b(j)$, then let us prove that for each $k$ neighbor of $j$ in $G$, $u(i) \leq u(k)$.	
		Let $k$ be a neighbor of $j$ in $G$.
		If $b(j) < b(k)$, then like in the previous case, we have $e(i) = b(j) < b(k) \leq e(k)$ and thus $u(i)<u(k)$.
		Otherwise, let us prove that $k \in X_{e(i)}$.
		We have $b(k) \leq b(j) = e(i)$ and since $k$ is a neighbor of $j$ then $e(i) = b(j) \leq e(k)$.
		Thus, $e(i) \in [b(k), e(k)]$ and thus $k \in X_{e(i)}$. 
		We remark that $i$ and $j$ are in $X_{b(j)} = X_{e(i)}$ and $c(i) = c(j)$. Then, the value $c(j)$ corresponds to the second cases in the definition and we have $\forall k \in X_{b(j)},$ $u(i) \leq u(k)$ (the only case where $u(i) = u(k)$ being when $i = k$).
	\end{itemize}
	
\end{proof}

\section{Proof of Lemma~\ref{lem_path_b}}

\begin{lemma}[ Lemma~\ref{lem_path_b} ]
	Let $h \in F(n,q)$.
	Let $G = \IG^*(h)$.
	If we have $c: [n] \to [s]$ and $u \in \Pi([n])$ such that $G,c,u$ have the same properties as in Lemma~\ref{lem_path_a}, then we have $\kappa(h,u) \leq s$.
\end{lemma}
\begin{proof}
	For all $j \in [n],$ let $v(j) := \{ k \in [n]\ |\  (k,j) \in E \}$.
	For all $i \in [n]$,  let $g_i: A^{|v(i)| } \to A$ such that $g_i \circ pr_{v(i)}= h_i$.
	In other words, $\forall x \in A^n$, $g_i(x_{v(i)}) = h_i(x)$.
	We know that such a function exists by definition of the interaction graph.
	Let $I_1, I_2, \dots, I_s$ a partition of $[n]$ such that $\forall \ell \in [s],\ I_\ell := \{ i \in [n]\ |\ c(i)= \ell \}$.
	For all $j \in [n],$ let $I(j) = I_{c(j)}$.
	Let $w = (n+1,n+2,\dots,n+s,u(1),\dots,u(n))$.
	Without loss of generality, let us say that $u$ is the canonical update schedule $(1,\dots,n)$.
	Let $f \in F(n+s,q)$ such that
	$\forall z = xy \in A^{n+s}$,
	\begin{itemize}
		\item $\forall i \in [n], f_{i}(z) = \begin{cases} h_i(x)$ if $\forall k \in v(i),\ u(i) \leq u(k) \\
		y_{c(i)} -
		\sum\limits_{j \in I(i) \text{ with } u(j) < u(i)  }  x_j -
		\sum\limits_{j \in I(i) \text{ with } u(i) < u(j)  }  h_j(x)  $ otherwise $ \end{cases}.$
		\item $\forall \ell \in [s], f_{n+\ell}(z) = \sum\limits_{j \in I_\ell }  h_j(x)$.
	\end{itemize}
	Let $y' = f^{w_1,\dots,w_s}(z)_{[n+1,n+s]} = (\sum\limits_{j \in I_1 }  h_j(x), \dots, \sum\limits_{j \in I_s }  h_j(x))$.
	Let us prove by induction that, being assumed that $[0] = \emptyset$, we have $\forall i \in [0,n],$ $f^{w_1,\dots,w_{s+i}}(z)_{[n]} = h^{[i]}(x)$.
	First $f^{w_1,\dots,w_{s}}(z)_{[n]} = (xy')_{[n]} = x = h^{\emptyset}(x)$.
	Next, let $i \in [n]$, let us suppose that $f^{w_1,\dots,w_{s+i-1}}(z)_{[n]} = h^{[i-1]}(x)$.
	Let $z' = x'y' = f^{w_1,\dots,w_{s+i-1}}(z)$.
	There are two cases.
	If $\forall k \in v(i),\ u(i) \leq u(k)$ then $f_i(z') = h_i(x')$.
	In this case we have, $x'_{v(i)} = x_{v(i)}$.
	Thus, $f_i(z') = h_i(x)$.
	Otherwise, we have $\forall j \in I(i)$ with $u(i) < u(j), \forall k \in v(j), u(i) < u(k)$.
	In other words, for each such $k$ we have $ x'_k = x_k$ and thus $x'_{v(j)} = x_{v(j)}$.
	Let $\ell = c(i)$.
	We have $f^i(z') = y'_\ell -
	\sum\limits_{j \in I_\ell \text{ with } u(j) < u(i)  }  x'_j -
	\sum\limits_{j \in I_\ell \text{ with } u(i) < u(j)  }  h_j(x')$

	We know that:
	\begin{itemize}
		\item $y'_\ell = \sum\limits_{j \in I_\ell }  h_j(x)$.
		\item $\forall j \in I_\ell \text{ with } u(j) < u(i),\ x'_j = h_j(x)$
		\item $\forall j \in I_\ell \text{ with } u(i) \leq u(j),\ h_j(x') = g_j(x'_{v(j)}) = g_j(x_{v(j)}) = h_j(x)$ (because $\forall k \in v(j),\ (k,j) \in E$, and then, by hypothesis of this lemma, $u(i) \leq u(k)$).
	\end{itemize}
	Thus,
	\begin{equation*}
	\begin{array}{l@{}l}
	f_{i}(z')
	&{}= y'_{\ell} - \sum\limits_{j \in I_{\ell} \text{ with } u(j) < u(i)  }  x'_j - \sum\limits_{j \in I_{\ell} \text{ with } u(i) < u(j)  }  h_j(x') \\
	&{} = \sum\limits_{j \in I_\ell }  h_j(x) - \sum\limits_{j \in I_{\ell} \text{ with } u(j) < u(i)  }  h_j(x) - \sum\limits_{j \in I_{\ell} \text{ with } u(i) < u(j)  }  h_j(x)   \\
	&{} = \sum\limits_{j \in I_\ell \text{ with } u(j) = u(i+1)  }  h_j(x)\\
	&{} = h_{i}(x)\\.
	\end{array}
	\end{equation*}
	As a result, in both case we have $f_i(z') = h_i(x)$.
	Moreover, $f^{w_1,\dots,w_{s+i}}(z)_{[n]} = f^{s+i}(z')_{[n]} = h^{[i+1]}(x)$ and by induction $\forall i \in [0,n],$ $f^{w_1,\dots,w_{s+i}}(z)_{[n]} = h^{[i]}(x)$.
	In particular, we have $f^w(z)_{[n]} = h(x)$ and then $\pr_{[n]} \circ f^w = h \circ \pr_{[n]}$. 
	Thus, $\kappa(h,w) \leq s $. 
	As a result, $\kappa^{min}(h) \leq s $.	
\end{proof}

\end{document}